\documentclass[a4paper
               ,10pt
               ,oneside
               ,onecolumn
               ,final
               ]
               {IEEEtran}

\usepackage[cmex10]{amsmath}
\usepackage{amssymb}  
\usepackage{amsthm}
\usepackage{amsfonts}
\usepackage{graphicx}
\usepackage{bbm}
\usepackage[T1]{fontenc}
\usepackage{lmodern}
\usepackage{url}
\usepackage{authblk}
\usepackage{tocloft}
\addtocontents{toc}{\vspace{-0.3cm}}

\theoremstyle{plain}               
\newtheorem{thm}{Theorem}[section]
\newtheorem{lem}{Lemma}[section]
\newtheorem{cor}{Corollary}[section]

\newtheorem{defn}{Definition}[section]
\newtheorem{exmp}{Example}[section]
\theoremstyle{remark}

\def\beq{\begin{equation}}
\def\eeq{\end{equation}}
\def\bq{\begin{quote}}
\def\eq{\end{quote}}
\def\ben{\begin{enumerate}}
\def\een{\end{enumerate}}
\def\bit{\begin{itemize}}
\def\eit{\end{itemize}}

\def\ra{\rightarrow}

\def\lb{\left(}
\def\rb{\right)}
\def\lset{\lbrace}
\def\rset{\rbrace}
\def\lk{\left\langle}
\def\rk{\right\rangle}
\def\l|{\left|}
\def\r|{\right|}
\def\lbr{\left[}
\def\rbr{\right]}
\def\i{\text{id}}
\newcommand\C{\mathbbm{C}}

\newcommand\R{\mathbbm{R}}
\newcommand\N{\mathbbm{N}}
\newcommand\M{\mathfrak{M}}
\newcommand\D{\mathfrak{D}}

\newcommand{\Tm}{\mathcal{T}}
\newcommand{\Rm}{\mathcal{R}}
\newcommand{\Sm}{\mathcal{S}}
\newcommand{\Em}{\mathcal{E}}
\newcommand{\Dm}{\mathcal{D}}
\newcommand{\Cm}{\mathcal{C}}
\newcommand{\Lm}{\mathcal{L}}

\newcommand{\Um}{\mathcal{U}}
\newcommand{\Vm}{\mathcal{V}}

\newcommand{\Ltimes}{\oplus}

\newcommand{\unitary}{\mathfrak{un}}
\newcommand{\chan}{\mathfrak{ch}}
\newcommand{\sgroup}{\mathfrak{sg}}

\newcommand{\chanSet}{\mathfrak{C}}

\author{Alexander M\"uller-Hermes, David Reeb, and Michael M. Wolf
\thanks{The authors are with the Department of Mathematics, Technical University Munich (Email: muellerh@ma.tum.de, reeb.qit@gmail.com, wolf@ma.tum.de).}%
\thanks{A.\ M\"uller-Hermes and M.\ Wolf acknowledge support from the CHIST-ERA/BMBF project CQC. M. Wolf is also supported by the Alfried Krupp von Bohlen und Halbach-Stiftung. D.\ Reeb is supported by the Marie Curie Intra-European Fellowship QUINTYL.}%
}

\date{}

\begin{document}

\title{Quantum Subdivision Capacities and Continuous-time Quantum Coding}

\maketitle

\begin{abstract}

Quantum memories can be regarded as quantum channels that transmit information through time without moving it through space. Aiming at a reliable storage of information we may thus not only encode at the beginning and decode at the end, but also intervene during the transmission -- a possibility not captured by the ordinary capacities in Quantum Shannon Theory. In this work we introduce capacities that take this possibility into account and study them in particular for the transmission of quantum information via dynamical semigroups of Lindblad form. When the evolution is subdivided and supplemented by additional continuous semigroups acting on arbitrary block sizes, we show that the capacity of the ideal channel can be obtained in all cases. If the supplementary evolution is reversible, however, this is no longer the case. Upper and lower bounds for this scenario are proven. Finally, we provide a continuous coding scheme and simple examples showing that adding a purely dissipative term to a Liouvillian can sometimes increase the quantum capacity.

\end{abstract}


\section{Introduction}
Inspired by its classical analog~\cite{Shannon1948}, a goal of quantum Shannon theory is to quantify the optimal rate of quantum information transmission using many instances of a quantum channel and a suitable encoding and decoding of the information. The physical scenario that one has in mind here is the transmission of quantum information over some noisy carrier of information such as a wire connecting two different points in space. The noise introduced by the carrier is then represented by a quantum channel. 

Another possible scenario is that of a quantum memory, where quantum information is stored and which suffers from noise due to coupling to some environment. Here the information is transmitted in time and not in space. Note that in physically relevant situations noise acts continuously in time on the quantum memory. This kind of noise is often modeled by a quantum dynamical semigroup of quantum channels $\Tm_t = e^{t\Lm}$~\cite{Lindblad1976} indexed by a time $t$ and each corresponding to the accumulated noise up to $t$. If we try to determine the optimal rate of quantum information transmission (i.e. of storing quantum information for some time $t$) in this setting, the usual quantum capacity $\mathcal{Q}(T_t)$~\cite[pp. 561]{Nielsen2007} of the whole channel $\Tm_t$ is not the appropriate limit. It only allows to optimize the initial encoding and the final decoding of the quantum information, but it does not take into account, that the quantum channel representing the noise is continuous in time and that the system stays at the same location during the transmission, which allows intervention throughout the transmission. During the storage we could for instance read the information from the memory at any earlier time and then write it into the memory again with another encoding, thereby subdividing the time-evolution. This could be repeated many times until the information is finally read from the memory. Another way to protect a quantum memory could be to engineer a control affecting the memory in a continuous way to protect the stored information. 
In some cases we may even add tailored dissipation or decoherence on top of the given one with the effect of enhancing the transmission rate. 

In this paper we introduce and investigate capacities quantifying the optimal rates of information storage in a quantum memory affected by continuous-time Markovian noise. To each noise Liouvillian $\Lm$ generating such a noisy time-evolution and any time $t\in\R^+$ we assign the ``quantum subdivision capacity'' $\mathcal{Q}_\chanSet(t\Lm)$ depending on a subset $\chanSet$ of quantum channels. This capacity is the supremum of optimal communication rates using the noisy time-evolution, when it may be subdivided into arbitrarily many parts and when we are allowed to apply coding channels from $\chanSet$ in between these parts. We study these capacities for different sets $\chanSet$, which capture operational restrictions, and prove capacity theorems for some of them. In the cases of $\chanSet=\chan$, i.e., when the intermediate channels can be arbitrary quantum channels (no restriction), and for $\chanSet=\sgroup^\ast$, i.e., for intermediate coding channels generated by time-dependent Liouvillians, we show that $\mathcal{Q}_\chanSet=\log(d)$ on a $d$-dimensional system, which is the capacity of an ideal channel. To prove this for $\chanSet=\sgroup^\ast$, we use the decoupling approach to quantum capacities~\cite{Dupuis2010} in order to show that one needs only a sublinear number of ancilla systems for the intermediate coding steps. After studying these two cases we consider the quantum subdivision capacity $\mathcal{Q}_\unitary$, i.e., where the intermediate coding channels are unitaries. In contrast to the previous two cases we prove that $\mathcal{Q}_\unitary(t\Lm)$ can become arbitrarily small as $t$ grows. On the positive side, we however show that $Q_\unitary$ is always strictly positive and thereby improves upon the usual quantum capacity in the generic case where the noise channel becomes entanglement breaking after some finite time.

In the second part of the paper, we introduce and investigate the continuous quantum capacities $\mathcal{Q}^\text{cont}_\chanSet(t\Lm)$ of the noise Liouvillian $\Lm$ at a time $t$, where $\chanSet$ denotes a fixed subset of time-dependent Liouvillians. These capacities quantify the optimal rate of information transmission over a system affected by continuous-time Markovian noise, when the possibility of superimposing a continuous control from $\chanSet$ is taken into account. For a physical relevant class of Liouvillians we introduce a continuous coding scheme implementing a continuous error correcting code with time-\emph{in}dependent quasi-local coding Liouvillians. Thereby we prove that the continuous quantum capacity is in general higher than the usual quantum capacity of the channel corresponding to the accumulated noise of the time-evolution, even if $\chanSet$ contains only time-\emph{in}dependent quasi-local Liouvillians.

Quantum subdivision capacities share some similarities to capacities for quantum relay channels~\cite{Savov2011,JinJing2012}, i.e., quantum channels describing the scenario of information transmission over a sequence of relay stations each of which with the power to decode and re-encode the information. Here we consider a similar scenario, with the major difference that the noise per subdivision is decreased when including more subdivisions and where the number of subdivisions can be modified by the communicating parties. We also study how operational restriction on the intermediate coding channels affect the capacity.  

Some of the possibilities captured by the continuous quantum capacities have been studied previously to some extent. Continuous formulations of quantum error correction have been studied for instance in ~\cite{Paz1998,Ahn2002,Sarovar2005,Oreshkov2007}, and the possibility to use engineered dissipative control to protect a quantum memory has been studied in ~\cite{Pastawski2009,Pastawski2011}. Limitations for information storage in a quantum memory of fixed size have been studied in~\cite{BenOr2013}. Our paper is to some extent inspired by these previous results. A related mechanism to prevent iterations of quantum channels from becoming entanglement-breaking has been studied in~\cite{DePasquale2012}. We are interested in the optimal rates of storage in a quantum memory and we introduce capacities quantifying these optimal rates, when the possibility of continuous control is taken into account. Thereby we formulate the problem of optimal information storage in the context of quantum Shannon theory, and this allows us to use techniques such as the decoupling approach~\cite{Winter2007,Dupuis2010} for continuous-time problems.

\section{Notations and preliminaries}
\label{sec:2}

We will denote the set of complex $d\times d-$ matrices by $\M_d := \M\lb\C^d\rb$, the set of $d-$dimensional density matrices, i.e., positive $d\times d-$matrices of trace 1, by $\D_d := \D\lb \C^d\rb$. The $d\times d-$identity matrix will be denoted by $\mathbbm{1}_d$. 

In formulas involving a large number of tensor factors some confusion can occur about which map acts on which system. Therefore we will sometimes introduce labels for different subsystems, although they might denote the same space, and use them for labeling maps and states. For example, we would write $\lb\i_{A'}\otimes\Tm^{A\rightarrow B}\rb\lb\rho^{A'A}\rb\in \M\lb \C^{d_{A'}}\otimes \C^{d_{B}}\rb$, for the application of the linear map $\Tm^{A\ra B}:\M_{d_A}\ra \M_{d_B}$ to the tensor factor labeled by A of $\rho\in \D\lb\C^{d_{A'}}\otimes \C^{d_{A}}\rb$, where the identity map is denoted by $\i_{A'}:\M_{d_{A'}}\ra\M_{d_{A'}}$ or as $\i_{d_{A'}}$ depending on the context. Note that we might have $d_A = d_{A'}$ and therefore $\M_{d_A}=\M_{d_{A'}}$. We will use the same label with a modification like $A,A',\tilde{A},...$, when dealing with matrix spaces of the same dimension, where this confusion might occur. A $d-$dimensional maximally entangled state will be written as $\omega^{A'A}\in \D\lb\C^{d_{A'}}\otimes \C^{d_{A}}\rb$, where $d_{A}=d_{A'}$ as described. Partial traces acting on a given state will be written by simply omitting the labels, that have been traced out. Thus we write $\rho^{A'}$ instead of $\text{tr}_A\lb\rho^{A'A}\rb$ for a state $\rho\in \D\lb\C^{d_{A'}}\otimes \C^{d_{A}}\rb$. For the partial trace as a map, we will write $\text{tr}_A$ or $\text{tr}_{d_A}$ depending on the context.

The set of completely positive and trace-preserving maps, i.e. quantum channels, mapping $\M_{d_A}$ to $\M_{d_B}$ will be denoted as $\chan(d_A,d_B)$ or as $\chan(d)$ if $d=d_A=d_B$. We will simply write $\chan$ to denote the set of arbitrary quantum channels. Here completely positive means, that the map $\i_{C}\otimes \Tm : \M_{d_C}\otimes \M_{d_A}\ra \M_{d_C}\otimes \M_{d_B}$ is positive for all dimensions $d_C$. Quantum channels can be characterized in different ways. One possibility, which we will use frequently is given by the Stinespring dilation theorem~\cite{Stinespring1955}: A linear map $\Tm:\M_{d_A}\ra \M_{d_B}$ is a quantum channel iff it can be written as $\Tm\lb\rho\rb = \text{tr}_E\lb V^{A\ra BE}\rho \lb V^{A\ra BE}\rb^\dagger\rb$ with an 'environment'-system labeled by $E$ and an isometry $V^{A\ra BE}:\C^{d_A}\ra\C^{d_B}\otimes \C^{d_E}$. Sometimes we will be interested in the state of the environment after applying a quantum channel. This can be done using the complementary channel~\cite{Devetak20052}. For a quantum channel $\Tm:\M_{d_A}\ra \M_{d_B}$ with Stinespring dilation $\Tm\lb\rho\rb = \text{tr}_E\lb V^{A\ra BE}\rho \lb V^{A\ra BE}\rb^\dagger\rb$ we define the complementary channel by tracing out the output system B of $\Tm$ instead of the environment, i.e. $\Tm^c\lb\rho\rb = \text{tr}_B\lb V^{A\ra BE}\rho \lb V^{A\ra BE}\rb^\dagger\rb$.      
 
In the following we will consider a continuous semigroup of quantum channels denoted by $\Tm_{t}:\M_d\ra\M_d$. Mathematically this is a family of quantum channels parametrized by a non-negative parameter $t\in\R^{+}$ such that $\Tm_s\circ\Tm_t = \Tm_{s+t}$ and $\Tm_0 = \i_d$ and such that the function $t\mapsto \Tm_t$ is continuous. We will call such a semigroup a quantum dynamical semigroup. Physically quantum dynamical semigroups describe Markovian evolutions in continuous time. It is well known~\cite{Lindblad1976}, that a continuous semigroup of linear, trace-preserving maps is completely positive, i.e. is a quantum dynamical semigroup, iff it is generated by a Liouvillian $\Lm :\M_d\ra \M_d$ of the form 
\begin{align}
\Lm\lb\rho\rb =& -i\lbr H,\rho\rbr + \sum^N_{k=1}\lb A_k\rho A^\dagger_k - \frac{1}{2}A^\dagger_k A_k\rho - \frac{1}{2}\rho A^\dagger_k A_k\rb
\label{equ:Liouvillian}
\end{align} 
for some Hermitian matrix $H\in \M_d$, which can be interpreted as a Hamiltonian, and Kraus operators $A_k\in\M_d$ for $k\in\lset 1,\ldots,N\rset$. 

The set of all quantum channels that can be written as $\Tm = e^{\Lm}$ for a Liouvillian $\Lm :\M_d\ra \M_d$ of the form (\ref{equ:Liouvillian}), is denoted by $\sgroup(d)$. When not specifying the dimension we will simply write $\sgroup$ to denote the set of all quantum channels $\Tm = e^{\Lm}$ for arbitrary Liouvillians $\Lm$. An example of quantum channels $\Tm\in \sgroup(d)$ are unitary channels of the form $\Tm(\rho) = U\rho U^\dagger$, which are generated by Liouvillians of the form $\Lm(\rho) = -i\lbr H,\rho\rbr$ for some Hermitian matrix $H\in \M_d$. We will denote the set of all unitary channels mapping $\M_d$ to $\M_d$ by $\unitary\lb d\rb$. It is clear that quantum channels $\Tm\in \sgroup$ are infinitely divisible~\cite{Wolf2008} as we have $\Tm = e^{\Lm} = \prod^k_{l=1} e^{\frac{\Lm}{k}}$ for all $k\in \N$. 

We will need certain norms as distance measures. For a matrix $\rho\in\M_d$ we define the trace-norm as $\left\Vert \rho\right\Vert_1 := \sum^{d}_{i=1}s (\rho)_i$, where $\lset s (\rho)_i\rset\subset\R^+$ denote the singular values of $\rho$. For maps $\Tm :\M_d\ra \M_d$ we use the induced norm given by $\left\Vert\Tm\right\Vert_{1\ra 1} := \sup_{\left\Vert\rho\right\Vert=1} \left\Vert\Tm(\rho)\right\Vert_1$ and a regularized version $\left\Vert\Tm\right\Vert_\diamond := \sup_{n\in\N}\left\Vert\i_n\otimes \Tm\right\Vert_{1\ra 1}$, which is the dual of the completely bounded norm ~\cite[p. 26]{Paulsen2003}. 

In the following we are interested in the transmission of quantum information through a system, undergoing a Markovian time-evolution. Therefore we will need some more results from quantum information theory:

\begin{defn}[\text{Quantum capacity }$\mathcal{Q}$, see \cite{Lloyd1997,Kretschmann2004}]
\label{defn:QuantCap}
The \textbf{quantum capacity} of a quantum channel $\Tm:\M_{d_A}\ra \M_{d_B}$ is defined as 
\begin{align*}
\mathcal{Q}\lb \Tm\rb := \sup\lset R\in \R^+ :\text{ R achievable rate}\rset
\end{align*}
where a rate $R\in\R^+$ is called achievable if there exist sequences $\lb n_\nu\rb^\infty_{\nu=1},\lb m_\nu\rb^\infty_{\nu=1}$ such that $R = \limsup_{\nu\ra \infty}\frac{n_\nu}{m_\nu}$ and 
\begin{align}
\inf_{\Em,\Dm}\left\Vert \i_2^{\otimes n_\nu} - \Dm\circ \Tm^{\otimes m_\nu}\circ \Em\right\Vert_\diamond \ra 0 \hspace*{0.3cm}\text{as } \nu\ra \infty .
\label{equ:limitQuantCap}
\end{align}
The latter infimum is over all encoding and decoding quantum channels $\Em:\M^{\otimes n_\nu}_{2}\ra \M^{\otimes m_\nu}_{d_A}$ and $\Dm:\M^{\otimes m_\nu}_{d_B}\ra \M^{\otimes n_\nu}_{2}$, see Fig. \ref{fig:StandardCodingScheme}. 
\end{defn} 

\begin{figure}
     \centering
     \includegraphics[width=0.35\textwidth]{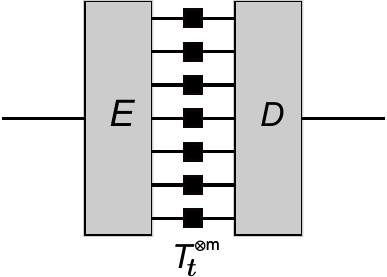}
     \caption{Coding scenario for the Quantum Capacity $\mathcal{Q}\lb\Tm\rb$ with encoding channel $\Em$ and decoding channel $\Dm$.}
	 \label{fig:StandardCodingScheme}
\end{figure}

We will need some theorems that have been proven for the quantum capacity. The first is the LSD-Theorem by Lloyd~\cite{Lloyd1997}, Shor~\cite{Shor2002} and Devetak~\cite{Devetak2005}, which expresses the quantum capacity through certain entropic quantities. In order to state this theorem we need some definitions. Note that here and in the following $\log$ will always mean logarithm with base 2. 

\begin{defn}[Coherent information, see~{\cite[p. 564]{Nielsen2007}}]
\label{defn:Entropies}
For a bipartite density matrix $\rho\in\D\lb\C^{d_A}\otimes\C^{d_B}\rb$ we define the \textbf{coherent information} of $\rho$ w.r.t. to the bipartition $A : B$ as 
\begin{align*}
I^{\text{coh}}\lb A\succ B\rb_\rho := S(\rho^B) - S(\rho^{AB}) =: -S\lb A|B\rb_\rho
\end{align*}  
where we denote by $S(\sigma) := -\text{tr}(\sigma\log(\sigma))$ the von-Neumann entropy and by $S(A|B)_\rho$ the \textbf{quantum conditional entropy of A given B} in the state $\rho$.

For a quantum channel $\Tm:\M_{d_A}\ra \M_{d_B}$ and a quantum state $\rho^{A'A}\in\D\lb\C^{d_{A'}}\otimes\C^{d_A}\rb$ we will write
\begin{align*}
I^{\text{coh}}\lb \rho,\Tm\rb := I^{\text{coh}}\lb A'\succ B\rb_{\lb\i_{A'}\otimes \Tm\rb\lb\rho_{A'A}\rb}.
\end{align*}
 
\end{defn}  

With this definition the following holds.

\begin{thm}[LSD, see ~\cite{Lloyd1997,Shor2002,Devetak2005}]
\label{thm:LSD}
For a quantum channel $\Tm:\M_{d_A}\ra \M_{d_B}$ we have 
\begin{align*}
\mathcal{Q}\lb \Tm\rb = \lim_{n\ra\infty} \frac{1}{n} \lbr \max_{\sigma^{A'A}} I^{\text{coh}}\lb \sigma ,\i_{A'}\otimes \Tm^{\otimes n}\rb\rbr. 
\end{align*}
The maximum is over states $\sigma^{A'A}\in \D\lb\C^{d_{A'}^n}\otimes \C^{d_{A'}^n}\rb$.
\end{thm}

Using this theorem and continuity of the coherent information it has been shown in ~\cite{Leung2009}, that the quantum capacity is continuous, or more specifically that the following holds.

\begin{thm}[Continuity of quantum capacity, see ~\cite{Leung2009}]
\label{thm:ContinuityQuantCap}
For quantum channels $\Tm,\Sm:\M_{d_A}\ra \M_{d_B}$ and $1\geq \epsilon >0$ with $\left\Vert \Tm - \Sm\right\Vert_\diamond\leq \epsilon$ we have
\begin{align*}
|\mathcal{Q}(\Tm) - \mathcal{Q}(\Sm)| \leq 8\epsilon d_B + 4 H(\epsilon)
\end{align*}
where $H(p) = -p\log(p) - (1-p)\log(1-p)$ denotes the binary entropy.
\end{thm}

The above theorem shows, that for every quantum dynamical semigroup $\Tm_t:\M_{d_A}\ra \M_{d_A}$ we have $\mathcal{Q}(\Tm_t)\ra \log(d_A)$ for $t\ra 0$, thus the capacity converges to its maximal possible value. For our argumentation continuity of the coherent information~\cite{Leung2009} will be sufficient. More specifically we will need that for every quantum dynamical semigroup $\Tm_t:\M_{d_A}\ra \M_{d_A}$ we have 

\begin{align}
I^{\text{coh}}\lb \omega^{A'A} ,\i_{A'}\otimes \Tm_t\rb \ra \log\lb d_A\rb\hspace*{0.3cm}\text{as }t\ra 0.
\label{equ:CoherentInfoToZero}
\end{align}  

\section{Quantum subdivision capacities}
\label{sec:3}

\subsection{Definition}

Consider a Liouvillian $\Lm:\M_d\ra\M_d$ of the form (\ref{equ:Liouvillian}). We want to define a capacity for the transmission of quantum information using a system that undergoes a noisy time-evolution generated by $\Lm$. The capacity will take into account, that we may interrupt the time-evolution at any point, i.e., that the quantum channel generated by $\Lm$ is infinitely divisible, and perform certain quantum operations before resuming the time-evolution. 

We will denote the set of quantum channels that are allowed to be applied in the intermediate steps by $\chanSet$. Formally this may be any subset $\chanSet\subseteq \chan$ and we will state our definition in the most general form for arbitrary $\chanSet$. Note that this definition yields a family of quantum capacities depending on the choice of $\chanSet$ and we will later study some relevant choices in more detail. 
 
For convenience we will include a time parameter $t\in\R^+$ in our definition of the capacity. This is because we want to emphasize the time-dependence of a quantum dynamical semigroup generated by a fixed Liouvillian $\Lm$. 

\begin{figure*}
     \centering
     \includegraphics[width=0.8\textwidth]{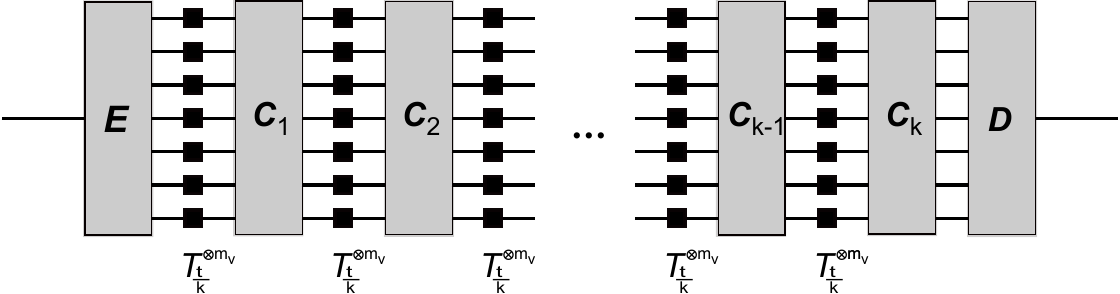}
     \caption{General coding scenario for the quantum subdivision capacity $\mathcal{Q}_\chanSet$. The intermediate coding channels $\Cm_l$ have to be from the set $\chanSet$ and $\Em,\Dm$ can be arbitrary quantum channels. Furthermore the number of subdivisions $k\in\N$ can be chosen freely.}
     \label{fig:subdivision}
\end{figure*} 

\begin{defn}[Quantum subdivision capacity $\mathcal{Q}_\chanSet$]
\label{defn:SubdivisionCap}
For $\chanSet\subseteq \chan$ the $\chanSet-$\textbf{quantum subdivision capacity} of a  Liouvillian $\Lm:\M_{d}\ra \M_{d}$ of the form (\ref{equ:Liouvillian}) at a time $t\in\R^+$ is defined as 
\begin{align*}
\mathcal{Q}_\chanSet\lb t\Lm\rb := \sup\lset R\in \R^+ :\text{ R achievable rate}\rset
\end{align*}
where a rate $R\in\R^+$ is called achievable if there exist sequences $\lb n_\nu\rb^\infty_{\nu=1},\lb m_\nu\rb^\infty_{\nu=1}$ such that $R = \limsup_{\nu\ra \infty}\frac{n_\nu}{m_\nu}$ and we have 
\begin{align}
\inf_{k,\Em,\Dm,\Cm_1,\ldots,\Cm_k} \left\Vert \i^{\otimes n_\nu}_2 - \Dm\circ \prod^k_{l=1} \lb\Cm_l\circ \Tm^{\otimes m_\nu}_{\frac{t}{k}}\rb\circ \Em\right\Vert_{\diamond} \ra 0
\label{equ:CodingSubdivisionCap}
\end{align}
as $\nu\ra \infty$. The latter infimum is over the number of subdivisions $k\in\N$ for which the channels $\Tm_{\frac{t}{k}} := e^{\frac{t}{k}\Lm}$ are defined, arbitrary encoding and decoding quantum channels $\Em:\M^{\otimes n_\nu}_{2}\ra \M^{\otimes m_\nu}_{d}$ and $\Dm:\M^{\otimes m_\nu}_{d}\ra \M^{\otimes n_\nu}_{2}$ and appropriate coding channels $\Cm_l\in \chanSet$ from the chosen subset, see Fig. \ref{fig:subdivision}.
\end{defn}

Different choices of $\chanSet$ in Definition \ref{defn:ContQuantCap} will in general lead to different quantities $\mathcal{Q}_\chanSet$. Note that by choosing $\chanSet=\lset \i_n : n\in\N\rset$ we have $\mathcal{Q}_\chanSet\lb t\Lm\rb = \mathcal{Q}\lb e^{t\Lm}\rb$, as we can only choose the identity map in the intermediate steps, and we recover the quantum capacity from Definition \ref{defn:QuantCap}. In the following we will consider three sets $\chanSet\subseteq \chan$ of physical relevance and the corresponding quantum subdivision capacities $\mathcal{Q}_\chanSet$. 

At first we will look at the case where arbitrary quantum channels are allowed in the intermediate coding steps. This corresponds to the choice $\chanSet = \chan$ for which we will obtain $\mathcal{Q}_\chan\lb t\Lm\rb = \log(d)$ for any Liouvillian $\Lm:\M_d\ra \M_d$ of the form (\ref{equ:Liouvillian}) and any $t\in\R^+$, see Theorem \ref{thm:CPTPSubCap}. Note that this is the maximal possible value on a $d$-dimensional system.  

In the second case, we will only allow quantum channels in $\chanSet$ which are composed of quantum dynamical semigroups to be applied in the intermediate steps, i.e., channels of the form $\prod^N_{i=1} e^{\Lm_i}$ for arbitrary $N\in\N$ and Liouvillians $\Lm_i$ of the form (\ref{equ:Liouvillian}). We will denote this set by $\sgroup^\ast$. The channels in $\sgroup^\ast$ can be thought of as generated from time-dependent Liouvillians and as their determinant is always positive they form a proper subset of $\chan$ ~\cite{Wolf2008}. Note also that $\sgroup^\ast\neq \sgroup$ ~\cite{Wolf20082}. For the quantum subdivision capacity we will again obtain the highest possible capacity $\mathcal{Q}_{\sgroup^\ast}\lb t\Lm\rb = \log(d)$ for any Liouvillian $\Lm:\M_d\ra \M_d$ of the form (\ref{equ:Liouvillian}) and any $t\in\R^+$, see Theorem \ref{thm:SemigroupSubCap}. The proof of this statement will turn out to be more involved as the set of possible coding channels is restricted.    

For the third case we constrain the set of coding operations further and only allow unitary channels to be applied in the intermediate steps, i.e., channels of the form $U\rho U^{\dagger}$ for a unitary matrix $U$. We will denote the set of unitary channels by $\unitary$. Here we will show, that there exist Liouvillians $\Lm:\M_d\ra\M_d$ such that $\mathcal{Q}_\unitary\lb t\Lm\rb\leq e^{-t}\log(d)$ which therefore becomes arbitrary small as $t\ra \infty$. We will derive some further lower bounds in this case.\newpage

Note that the capacities defined above share some similarities with quantum capacities of quantum serial relays, see for instance ~\cite{Savov2011,JinJing2012} and the references therein. The major difference here is that, while for quantum relays the noise channels acting between the relays are fixed quantum channels, we have the additional freedom to subdivide the noise channel into ever smaller time-steps while keeping the total time fixed. We also restrict the allowed intermediate coding operations to be from the chosen subset $\chanSet\subseteq \chan$ and we show that the value of $\mathcal{Q}_\chanSet\lb t\Lm\rb$ will depend on this choice.

\subsection{Quantum subdivision capacity with general add-ons}

We will start with the widest case $\chanSet = \chan$ of arbitrary coding quantum channels in the intermediate steps, see Definition \ref{defn:SubdivisionCap}.

\begin{thm}[$\chan$-quantum subdivision capacity]
\label{thm:CPTPSubCap}
For any Liouvillian $\Lm:\M_d\ra \M_d$ of the form (\ref{equ:Liouvillian}) and any $t\in\R^+$ we have
\begin{align*}
\mathcal{Q}_{\chan}\lb t\Lm\rb = \log(d)\,.
\end{align*} 
\end{thm}
\begin{proof}
By continuity, Theorem \ref{thm:ContinuityQuantCap}, we have $\mathcal{Q}\lb \Tm_{\frac{t}{k}}\rb\ra \mathcal{Q}\lb \i_d \rb = \text{log}\lb d\rb$ for quantum channels $\Tm_{\frac{t}{k}} := e^{\frac{t}{k}\Lm}$ as $k\ra \infty$. 
Thus for any $\delta>0$ we can choose a $K\in\N$ such that $\mathcal{Q}\lb \Tm_{\frac{t}{K}}\rb > \text{log}\lb d\rb-\delta$. By Definition \ref{defn:QuantCap} there are sequences $\lb n_\nu\rb^\infty_{\nu=1}, \lb m_\nu\rb^\infty_{\nu=1}$ such that $\text{limsup}_{\nu\ra \infty}\frac{n_\nu}{m_ \nu} \geq \text{log}\lb d\rb-\delta$ and 
\begin{align}
 \left\Vert \i^{\otimes n_\nu}_2 - \tilde{\Dm}_\nu\circ \Tm^{\otimes m_\nu}_{\frac{t}{K}}\circ \tilde{\Em}_\nu\right\Vert_{\diamond} \ra 0 \hspace*{0.3cm}\text{ as }\nu\ra\infty
\label{Equ:Limit}
\end{align}
holds for encoding and decoding quantum channels $\tilde{\Em}_\nu:\M^{\otimes n_\nu}_{2}\ra \M^{\otimes m_\nu}_{d_A}$ and $\tilde{\Dm}_\nu:\M^{\otimes m_\nu}_{d_B}\ra \M^{\otimes n_\nu}_{2}$. Now choose the coding maps in Definition \ref{defn:SubdivisionCap} as $\tilde{\Cm}^\nu_k=\i_{d^{m_\nu}}$ and $\tilde{\Cm}^\nu_l=\tilde{\Em}_\nu\circ\tilde{\Dm}_\nu$ for $l\leq k-1$, see Fig. \ref{fig:CPTPCodingScheme}. 

\begin{figure*}
     \centering
     \includegraphics[width=0.8\textwidth]{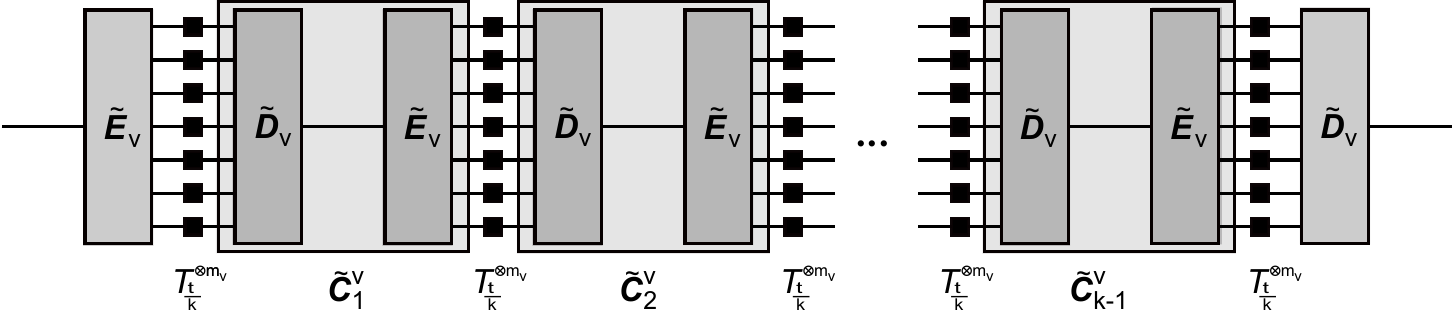}
     \caption{Subdivision Coding Scheme for the proof of Theorem \ref{thm:CPTPSubCap}. For any $\delta>0$ choose $k\in\N$ such that a rate $R > \log(d)-\delta$ is achievable for the usual capacity $\mathcal{Q}$. In each of the $k$ steps the local error satisfies $\|\i^{\otimes n_\nu}_{2} -\tilde{\Dm_\nu}\circ\Tm^{\otimes m_\nu}_{t/k}\circ\tilde{\Em_\nu}\|_\diamond\ra 0$ for encoding and decoding channels $\tilde{\Em_\nu},\tilde{\Dm_\nu}$ achieving the rate $R$. As $k$ is independent of $\nu$ the total error vanishes in the limit $\nu\ra\infty$ showing that $R$ is achievable for the quantum subdivision capacity.}
     \label{fig:CPTPCodingScheme}
\end{figure*}

Inserting these special coding maps into (\ref{equ:CodingSubdivisionCap}) leads to 
\begin{align*}
\inf_{k,\Em,\Dm,\Cm^\nu_1,\ldots,\Cm^\nu_k} &\left\Vert \i^{\otimes n_\nu}_2 - \Dm\circ\prod^k_{l=1} \lb \Cm^\nu_l\circ\Tm^{\otimes m_\nu}_{\frac{t}{k}}\rb\circ \Em\right\Vert_{\diamond} \leq  \left\Vert \i^{\otimes n_\nu}_2 - \prod^K_{l=1}\lb \tilde{\Dm}_\nu\circ \Tm^{\otimes m_\nu}_{\frac{t}{K}}\circ \tilde{\Em}_\nu\rb\right\Vert_{\diamond} \\
& = \bigg{\Vert} \i^{\otimes n_\nu}_2 - \tilde{\Dm}_\nu\circ \Tm^{\otimes m_\nu}_{\frac{t}{K}}\circ \tilde{\Em}_\nu + \tilde{\Dm}_\nu\circ \Tm^{\otimes m_\nu}_{\frac{t}{K}}\circ \tilde{\Em}_\nu\circ\lb \i^{\otimes n_\nu}_2 - \prod^{K-1}_{l=1}\lb \tilde{\Dm}_\nu\circ \Tm^{\otimes m_\nu}_{\frac{t}{K}}\circ \tilde{\Em}_\nu\rb\rb\bigg{\Vert}_{\diamond} \\
&\leq \left\Vert \i^{\otimes n_\nu}_2 - \tilde{\Dm}_\nu\circ \Tm^{\otimes m_\nu}_{\frac{t}{K}}\circ \tilde{\Em}_\nu\right\Vert_{\diamond} + \left\Vert \i^{\otimes n_\nu}_2 - \prod^{K-1}_{l=1}\lb \tilde{\Dm}_\nu\circ \Tm^{\otimes m_\nu}_{\frac{t}{K}}\circ \tilde{\Em}_\nu\rb\right\Vert_{\diamond} \\
&\vdots \\
&\leq K \left\Vert \i^{\otimes n_\nu}_2 - \tilde{\Dm}_\nu\circ \Tm^{\otimes m_\nu}_{\frac{t}{K}}\circ \tilde{\Em}_\nu\right\Vert_{\diamond} \rightarrow 0 \hspace*{0.3cm}\text{as }\nu\ra\infty.
\end{align*}
For the above chain of inequalities we used the triangle inequality, submultiplicativity of $\left\Vert\cdot\right\Vert_\diamond$, the fact that $\left\Vert \Tm\right\Vert_\diamond = 1$ for any quantum channel $\Tm$ and finally the limit (\ref{Equ:Limit}) from the usual quantum capacity. This shows that any rate $R<\log(d)$ is achievable, which implies that $\mathcal{Q}_{\chan}\lb t\Lm\rb = \log(d)$ for all $\Lm$ and $t\in\R^+$.

\end{proof}

\section{Quantum subdivision capacity with continuous semigroup add-ons}
\label{sec:4}

In this section we will consider the subdivision capacity from Definition \ref{defn:SubdivisionCap} for the case $\chanSet = \sgroup^\ast$, which corresponds to the use of arbitrary quantum channels composed of quantum dynamical semigroups.

\begin{thm}[$\sgroup^\ast$-quantum subdivision capacity]
\label{thm:SemigroupSubCap} 
For any Liouvillian $\Lm:\M_{d}\ra \M_{d}$ of the form (\ref{equ:Liouvillian}) and any $t\in\R^+$ we have
\begin{align*}
\mathcal{Q}_{\sgroup^\ast}\lb t\Lm\rb = \log(d)\,.
\end{align*}

\end{thm}

Note that the quantum channels depolarizing towards pure states $\Dm_t:\M_{d}\ra \M_{d}$ given as
\begin{align*}
\Dm_t\lb\rho\rb = (1-e^{-rt})\text{tr}\lb\rho\rb\l| \phi\rk\lk \phi\r| + e^{-rt} \rho 
\end{align*}
for an arbitrary pure state $\l| \phi\rk\lk \phi\r|\in \D_{d}$ and a rate $r\in\R^+$ are contained in $\sgroup^\ast$. The idea to prove Theorem \ref{thm:SemigroupSubCap} is to use these channels to obtain almost pure states. The almost pure states obtained this way are then used to implement the optimal encoding-decoding operations used in the proof of Theorem \ref{thm:CPTPSubCap}, see Fig. \ref{fig:CPTPCodingScheme}, via unitary quantum channels, which are contained in $\sgroup$. As $\sgroup^\ast$ is closed under composition, we may compose depolarizing channels and unitary channels yielding overall quantum channels still contained in $\sgroup^\ast$ and therefore valid coding channels for the intermediate coding steps.  

Note that Definition \ref{defn:SubdivisionCap} of the quantum subdivision capacity requires, that the generation of these pure state ancillas has to occur within the coding scheme (\ref{equ:CodingSubdivisionCap}). This is not the same as having pure ancillas for free. Every ancilla system leads in the present context to an increase of channel copies required to carry it through the coding scheme. It is a priori not clear how many ancilla systems will be needed to implement the coding scheme from Theorem \ref{thm:CPTPSubCap} using unitary channels and how this number scales with $\nu$ in the limit (\ref{equ:CodingSubdivisionCap}). We will prove that a sublinear number of pure ancillas is sufficient. 

The following simple lemma will be needed for the proof:  

\begin{lem}[Implementing isometries via unitaries]
\label{lem:ImplementIsometry}
For $k,n\in\N$, any isometry $V:\C^{n}\ra \C^{nk}$ can be written as 
\begin{align*}
V\l|\psi\rk = U\lb \l|\psi\rk\otimes \l|\phi\rk\rb
\end{align*}
with a unitary $U:\C^{nk}\ra\C^{nk}$ and an auxiliary state $\l|\phi\rk\in\C^{k}$
\end{lem}

\begin{proof}

We can write the isometry $V = \sum^{n}_{i=1} \l| v_i\rk\lk i\r|$ for some orthonormal basis $\lset \l| i\rk\rset_{i}\subset \C^{n}$ and an orthonormal system $\lset \l| v_i\rk\rset_{i}\subset \C^{nk}$. Now fix an auxiliary state $\l|\phi\rk\in \C^{k}$ to obtain an orthonormal system $\lset \l| i\rk\otimes\l| \phi\rk\rset_{i}\subset\C^{nk}$. Since all orthonormal systems of the same size are related by some unitary transformation $U$, we are finished.
\end{proof}

Before we state the proof of the coding Theorem \ref{thm:SemigroupSubCap}, we also need a technical lemma, which we will use to quantify the number of ancillas needed to implement the optimal coding schemes achieving rates as required in the statement of the theorem. The proof of the lemma will be done using a standard technique from the method of types ~\cite{Winter1999}.     

\begin{lem}[Approximate Choi matrix purification with small environment]
\label{lem:ChoiMatrSmallEnv}
Let $\Tm:\M_{d_A}\ra \M_{d_B}$ be a quantum channel with Stinespring isometry $V^{A\ra BE}:\C^{d_A}\ra\C^{d_B}\otimes\C^{d_E}$ and consider a purification of the Choi matrix of $\Tm$ given by 
\begin{align*}
\l|\sigma^{A'BE}\rk &\lk\sigma^{A'BE}\r| = \lb \mathbbm{1}_{d_{A'}}\otimes V^{A\ra BE}\rb \omega^{A'A} \lb \mathbbm{1}_{d_{A'}}\otimes V^{A\ra BE}\rb^\dagger
\end{align*}
with a system labeled by $A'$ of size $d_{A'}=d_A$.
Then for arbitrary $\delta > 0$ and all sufficiently large $m\in\N$, there exist pure states $\l|\tilde{\sigma}^{A'BE}\rk\lk\tilde{\sigma}^{A'BE}\r|\in \D\lb\C^{d^m_{A'}}\otimes\C^{d^m_B}\otimes \C^{d^m_E}\rb$ such that 
\begin{align*}
&\left\Vert \l|\tilde{\sigma}^{A'BE}\rk\lk\tilde{\sigma}^{A'BE}\r| - \lb\l|\sigma^{A'BE}\rk\lk\sigma^{A'BE}\r|\rb^{\otimes m}\right\Vert_1 \leq 2^{1 -\frac{mc'\delta^2}{2}}
\end{align*} 
and $\text{rank}\lb \tilde{\sigma}^{E}\rb \leq 2^{m S\lb \sigma^{E}\rb + cm\delta}$ for some constants $c,c'\in\R^+$ independent of $\delta$ and $m$.

\end{lem}

\begin{proof}

By the method of typical subspaces, see e.g. ~\cite{Devetak2005,Bennett2012,Winter2007}, there exists a projector $\Pi^E_\delta:\C^{d_E^m}\ra\C^{d_E^m}$ onto a subspace $E_\delta\subset \C^{d_E^m}$ with dimension $\dim\lb E_\delta\rb \leq 2^{m S\lb \sigma^E\rb + cm\delta}$ for some constant $c\in\R^+$ which does not depend on $\delta$ or $m$, such that the statements in the lemma are fulfilled for the pure states
\begin{align*}
\l|\tilde{\sigma}^{A'BE}\rk = \frac{\lb \mathbbm{1}_{d^m_{A'}}\otimes \mathbbm{1}_{d^m_B}\otimes \Pi^E_\delta\rb\l|\sigma^{A'BE}\rk^{\otimes m}}{\text{tr}\lb \Pi^E_\delta \lb\sigma^E\rb^{\otimes m}\rb}
\end{align*}

\end{proof}

With the lemmata in place we can prove the capacity theorem stated above.

\begin{proof}$\lb\text{of Theorem \ref{thm:SemigroupSubCap}} \rb$

We will show that $\tilde{R}=\frac{\log\lb d\rb - \delta}{1+\delta m}$ is an achievable rate as defined in Definition \ref{defn:SubdivisionCap} for any $\delta>0$ and some constant $m\in\R^+$, which will be specified in the proof. For $R= \text{log}\lb d\rb-\delta $, consider the sequences $n_\nu = R\nu$ and $m_\nu = \nu$ and note that we will for brevity write $n_\nu = R\nu$ instead of $n_\nu = \lfloor R\nu\rfloor$ and omit the floor operations, which are always implicitly assumed, when we are talking about integer sequences.

To construct a quantum subdivision coding scheme achieving the rate $\tilde{R}=\frac{\log\lb d\rb - \delta}{1+\delta m}$, we start as in the proof of Theorem \ref{thm:CPTPSubCap} and consider a coding scheme for the usual capacity, Definition \ref{defn:QuantCap}, for a channel of the form $\Tm_{t/k}:= e^{\frac{t}{k}\Lm}$ achieving the rate $R$ close to the maximal value $\log(d)$. Then we use this 'local' scheme to build the total coding scheme for the subdivision capacity similar to the previous proof. Compare Fig. \ref{fig:CPTPCodingScheme} to Fig. \ref{fig:SGCodingScheme} to get further insight into the intuition behind this proof.

By continuity $(\ref{equ:CoherentInfoToZero})$ there exists a $K\in\N$ such that $I^{\text{coh}}\lb \omega^{A'A} ,\i_{A'}\otimes \Tm_{t/K}\rb  > \text{log}\lb d\rb-\delta$ for the quantum channel $\Tm_{t/K}:= e^{\frac{t}{K}\Lm}$. Thus the rate $R= \text{log}\lb d\rb-\delta $ is achievable and we can use the decoupling approach, see Appendix A, to construct a coding scheme.   

By Lemma \ref{lem:RateAchievDecoup} there is a sequence $\epsilon_\nu\ra 0$ for $\nu\ra\infty$ and there exist unitary channels $\Um\in \unitary\lb {d}^\nu\rb$ and isometric embeddings $\mathcal{V}^{R\ra A}_\nu : \M_{2^{R\nu}}\ra \M_{{d}^\nu} $ for all $\nu\in\N$ such that for all states $\rho^{R'R}\in\D\lb\C^{2^{R\nu}}\otimes \C^{2^{R\nu}}\rb$  
\begin{align}
\bigg{\Vert} \lb\i_{R'}\otimes \lbr\lb\Tm_{t/K}^c\rb^{\otimes\nu}\circ \Um\circ \mathcal{V}^{R\ra A}_\nu\rbr\rb\lb \rho^{R'R}\rb -\rho^{R'}\otimes \lb\sigma^E\rb^{\otimes\nu}\bigg{\Vert}_1 \leq \epsilon_\nu 
\label{equ:Decoup}
\end{align}
is fulfilled. We denote by $\sigma^E$ the reduced density matrix of the purified Choi matrix corresponding to the channel $\Tm_{t/K}= e^{\frac{t}{K}\Lm}$ given by
\begin{align*}
\l|\sigma^{A'AE}\rk&\lk\sigma^{A'AE}\r| = \lb \mathbbm{1}_{A'}\otimes V^{A\ra AE}\rb \omega^{A'A} \lb \mathbbm{1}_{A'}\otimes V^{A\ra AE}\rb^\dagger
\end{align*}  
where as before $ V^{A\ra AE}:\C^{d}\ra \C^{d}\otimes \C^{d_E}$ denotes the Stinespring isometry of the quantum channel $\Tm_{t/K}$.

As described in Appendix A we can use the encoding maps $\tilde{\Em}_\nu : \M_{2^{R\nu}}\ra \M_{{d}^\nu}$ given by 
\begin{align}
\tilde{\Em}_\nu\lb\rho\rb = \Um\circ \mathcal{V}^{R\ra A}_\nu\lb\rho\rb
\label{equ:encoding}
\end{align}
to define a coding scheme for the channel $\Tm_{t/K}$, see Definition \ref{defn:QuantCap}. By Lemma \ref{lem:ImplementIsometry} the encoding operation (\ref{equ:encoding}) can be implemented as
\begin{align}
\tilde{\Em}_\nu\lb\rho\rb = \tilde{U}_\nu\lb \rho\otimes \l|\phi_\nu\rk\lk \phi_\nu\r|\rb \tilde{U}_\nu^\dagger
\label{equ:EncodUni}
\end{align}
with a unitary $\tilde{U}_\nu:\C^{d^\nu}\ra \C^{d^\nu}$ and a pure states $\l|\phi_\nu\rk\in\C^{2^{\delta\nu}}$.

So far we have implemented the encoding operation via unitaries. These encoding operations achieve decoupling in the limit $\nu\ra\infty$ as explained in Appendix A. Now we will do the same for the decoding operations of the 'local' coding scheme, see Lemma \ref{lem:GeneralDecoding}. For reasons that will become clear in the further discussion, we will employ Lemma \ref{lem:ChoiMatrSmallEnv} to decrease the dimension of the ancilla system needed for implementing the decoding operation via unitaries. \newpage

Using the Lemma \ref{lem:ChoiMatrSmallEnv} there is a pure state $\l|\tilde{\sigma}^{A'AE}\rk\in\C^{d^\nu_{A'}}\otimes\C^{d^\nu_B}\otimes \C^{d^\nu_E}$ such that 
\begin{align*}
\left\Vert \l|\tilde{\sigma}^{A'AE}\rk\lk\tilde{\sigma}^{A'AE}\r| - \lb\l|\sigma^{A'AE}\rk\lk\sigma^{A'AE}\r|\rb^{\otimes \nu}\right\Vert_1 \leq 2^{1 -\frac{\nu c'\delta^2}{2}}
\end{align*} 
and such that $\text{rank}\lb \tilde{\sigma}^{E}\rb\leq 2^{\nu S\lb \sigma^E\rb + c\nu\delta}$ for some constants $c,c'\in\R^+$ independent of $\nu$ and $\delta$. 
Inserting this approximate state into the decoupling bound (\ref{equ:Decoup}) shows that 
\begin{align*}
&\left\Vert \lb\i_{R'}\otimes \lbr\lb\Tm_{t/K}^{c}\rb^{\otimes\nu}\circ \Em_\nu\rbr\rb\lb \rho^{R'R}\rb - \rho^{R'}\otimes \tilde{\sigma}^{E}\right\Vert_1 \leq \lbr\epsilon_\nu + 2^{1 -\frac{\nu c'\delta^2}{2}}\rbr \ra 0 
\end{align*}
as $\nu\ra\infty$. Therefore using the above encoding operations $\Em_\nu$ we achieve decoupling for environment states $\tilde{\sigma}^E$ and we can construct decoding quantum channels $\tilde{\Dm}_\nu:\M^{\otimes \nu}_{d}\ra \M^{\otimes R\nu}_{2}$ as shown in Lemma \ref{lem:GeneralDecoding}. These decoding operations may be assumed to be of the form
\begin{align}
\tilde{\Dm}_\nu\lb\rho\rb = \text{tr}_{d_{E^\nu_\delta}} W_\nu\rho W_\nu^\dagger
\label{equ:Decoding}
\end{align}
for isometries $W_\nu:\C^{d^\nu}\ra \C^{2^{R\nu}}\otimes \C^{d_{E^\nu_\delta}}$, where $d_{E^\nu_\delta} \leq 2^{\nu S\lb \sigma^E\rb + c\nu\delta}$. To determine the dimension of the ancilla system we use the assumptions made above and get
\begin{align*}
\log\lb d\rb - \delta &< I^{\text{coh}}\lb \omega^{A'A} ,\i_{A'}\otimes \Tm_{t/K}\rb  \\&= S\lb\sigma^A\rb - S\lb \sigma^{A'A}\rb \\
& \leq \log\lb d\rb - S\lb \sigma^{E}\rb.
\end{align*}
Finally we obtain $S\lb \sigma^E\rb < \delta$ and therefore $d_{E^\nu_\delta}< 2^{\nu \delta \lb 1+c\rb}$. By Lemma \ref{lem:ImplementIsometry} we can implement the decoding operation $\tilde{\Dm}_\nu$, see equation (\ref{equ:Decoding}), as 
\begin{align}
\tilde{\Dm}_\nu\lb\rho\rb = \text{tr}_{d_{E^\nu_\delta}} \lb U_\nu \lbr\rho\otimes\l|\tilde{\phi}_\nu\rk\lk \tilde{\phi}_\nu\r|\rbr U^\dagger_\nu\rb
\label{equ:DecodUni}
\end{align}
with a unitary $U_\nu\in \unitary\lb d^\nu 2^{\nu\delta c }\rb$ and a pure state $\l|\tilde{\phi}_\nu\rk\in \C^{2^{\nu\delta c }}$, where we made the auxiliary system a bit larger than necessary.
 
So far we constructed a coding scheme, see Definition \ref{defn:QuantCap}, for the quantum channel $\Tm_{t/k}=e^{\frac{t}{k}\Lm}$, which achieves the rate $R=\log(d)-\delta$. Furthermore we have shown how to implement the coding operations of the scheme using unitaries and pure ancilla states, see (\ref{equ:EncodUni}) and (\ref{equ:DecodUni}).  

Now we construct a subdivision coding scheme, see Definition \ref{defn:SubdivisionCap}, by concatenating the 'local' scheme $K$-times, basically in the same way as in the proof of Theorem \ref{thm:CPTPSubCap}. The only difference is, that we have to implement the coding scheme using maps from $\sgroup$ instead of arbitrary coding channels. Therefore we divide the system in the intermediate steps of equation (\ref{equ:CodingSubdivisionCap}) into two parts. The first part is the information carrying system of dimension $d^\nu$, in which the information will be encoded at a rate $R = \log(d)-\delta$. The second part is an auxiliary system of size $d^{\nu\delta m} = 2^{\nu\delta c}$ for $m:=  c\log_{d}(2)$. This system holds the ancilla states used to implement the decoding maps $\tilde{\Dm}_\nu$ acting on the information carrying system via unitaries acting on the full system. We will also show, that the auxiliary states used to implement the encoding maps $\tilde{\Em}_\nu$ via unitaries can be generated within the coding scheme without enlarging the auxiliary system, see Fig. \ref{fig:SGCodingScheme}.  

\begin{figure*}
\center
     \includegraphics[width=0.8\textwidth]{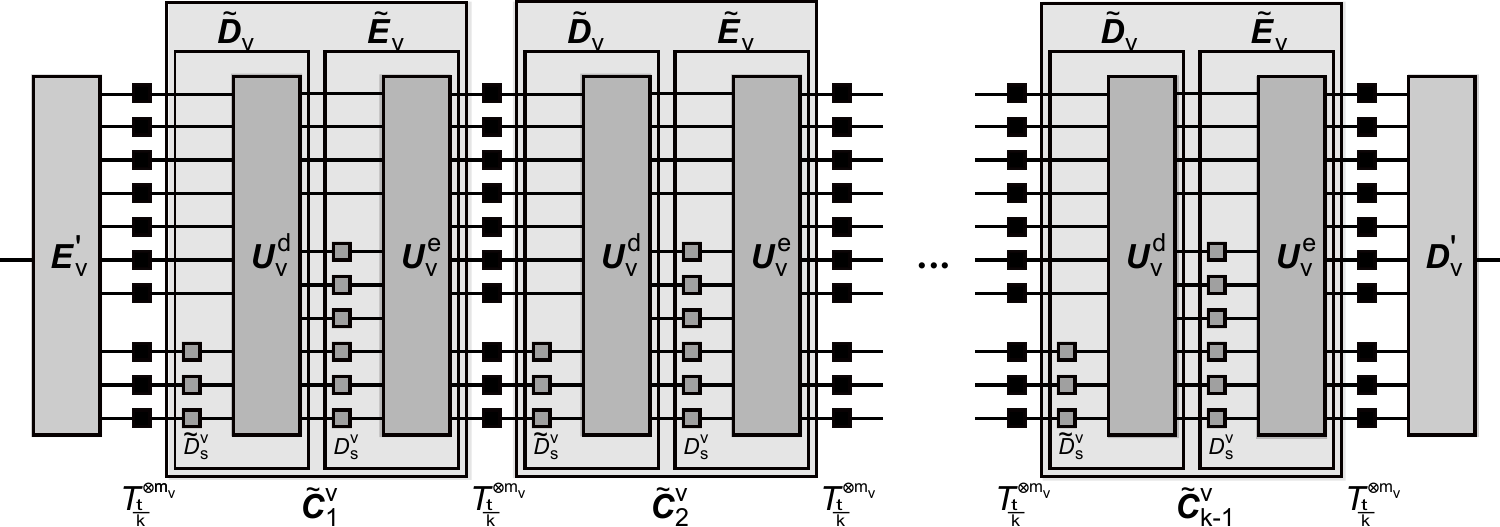}
     \caption{Subdivision Coding Scheme for the proof of Theorem \ref{thm:SemigroupSubCap}. Implement the subdivision coding scheme from Theorem \ref{thm:CPTPSubCap} via almost pure ancillas generated within the coding scheme using depolarizing channels $\Dm^\nu_{s}$ and unitary channels $\Um^e_\nu, \Um^d_\nu$. This requires some overhead depicted by a gap between the information carrying system and the auxiliary system producing the almost pure ancillas.}
     \label{fig:SGCodingScheme}
\end{figure*}

The rate of the scheme constructed this way can be seen to be $\tilde{R}=\frac{R}{1+\delta m} = \frac{\log\lb d\rb - \delta}{1+\delta m}$. The numerator is the rate with which the information is encoded in the information carrying system and the denominator is increased by $\delta m$, which corresponds to the additional channel uses needed to carry the auxiliary system. 

We have to define the subdivision coding scheme described above properly and show that the error, see equation (\ref{equ:CodingSubdivisionCap}), vanishes for $\nu\ra\infty$. Note that we hold the number of subdivisions $K$ in Definition \ref{defn:SubdivisionCap} fixed since the beginning of the proof. In the following take sequences $\lb n_\nu\rb_\nu$ and $\lb m_\nu\rb_\nu$, see Definition \ref{defn:SubdivisionCap}, defined as $n_\nu :=  R\nu$ and $m_\nu := \nu\lb 1+\delta m\rb$. Also note that the size of the auxiliary system $d^{\nu\delta m}$ is enough to hold the auxiliary states $\l|\tilde{\phi}_\nu\rk$ used to implement the decoding operations, see (\ref{equ:DecodUni}), on the information carrying system. By inspecting Equation (\ref{equ:DecodUni}) it is clear, that after applying the decoding operation a system of size $d_{E^\nu_\delta} = 2^{\nu\delta (1+c)}$ will be free from information. This system is large enough to contain the pure state $\l|\phi_\nu\rk\in\C^{2^{\nu\delta}}$ used to implement the encoding operation via a unitary, see (\ref{equ:EncodUni}). \newpage
Thus we can assume 
\begin{align}
\tilde{\Em}_\nu\lb\cdot\rb =  \Um^e_\nu \lb\cdot\otimes\l|\phi^e_\nu\rk\lk \phi^e_\nu\r|\rb 
\label{equ:Enc}
\end{align}
for unitary maps $\Um_\nu^e: \M_{2^{R\nu}}\otimes \M_{2^{\nu\delta (1+c)}}\ra  \M_{d^\nu}\otimes \M_{d^{\nu\delta m}}$ and pure auxiliary states $\l|\phi^e_\nu\rk\in \C^{2^{\nu\delta (1+c)}}$, which just extends the states $\l|\phi_\nu\rk$ from above. We may also assume
\begin{align}
\tilde{\Dm}_\nu\lb\cdot\rb = \text{tr}_{d^{\nu\delta m}} \lb \Um^d_\nu \lb\cdot\otimes\l|\tilde{\phi}_\nu\rk\lk \tilde{\phi}_\nu\r|\rb\rb 
\label{equ:Dec}
\end{align}
for unitary maps $\Um_\nu^d: \M_{d^\nu}\otimes \M_{d^{\nu\delta m}}\ra  \M_{2^{R\nu}}\otimes \M_{2^{\nu\delta (1+c)}}$. We can use the corresponding depolarizing channels $\Dm^\nu_{s_\nu}:\M_{2^{\nu\delta (1+c)}}\ra \M_{2^{\nu\delta (1+c)}}$ given by 
\begin{align*}
\Dm^\nu_{s_\nu}\lb\rho \rb = \lb 1-e^{-s_\nu}\rb\text{tr}\lb\rho\rb\l|\phi^e_\nu\rk\lk \phi^e_\nu\r| + e^{-s_\nu}\rho
\end{align*}
and $\tilde{\Dm}^\nu_{s_\nu}:\M_{d^{\nu\delta m}}\ra \M_{d^{\nu\delta m}}$ given by 
\begin{align*}
\tilde{\Dm}^\nu_{s_\nu}\lb\rho \rb = \lb 1-e^{-s_\nu}\rb\text{tr}\lb\rho\rb\l|\tilde{\phi}_\nu\rk\lk \tilde{\phi}_\nu\r| + e^{-s_\nu}\rho
\end{align*}
to obtain auxiliary states for implementing the encoding and decoding maps (\ref{equ:Enc}) and (\ref{equ:Dec}) with arbitrary accuracy. Note that these depolarizing channels can be defined as tensor products of local depolarizing channels, by choosing the unitaries in (\ref{equ:EncodUni}) and (\ref{equ:DecodUni}) accordingly. This is also depicted in Fig. \ref{fig:SGCodingScheme}. 

Define the coding maps $\lb\tilde{\Cm}^\nu_l\rb^{K-1}_{l=1}\subset \sgroup^\ast$ as 
\begin{align*}
\tilde{\Cm}^{\nu}_l = \Um^e_\nu \circ \lb\i_{2^{R\nu}}\otimes \Dm^\nu_s\rb \circ \Um^d_\nu\circ\lb\i_{d^{\nu}}\otimes \tilde{\Dm}^\nu_s\rb
\end{align*}
and $\tilde{\Cm}^\nu_K = \i_{d^{\nu\lb 1+\delta m\rb}}$.
Note that we can choose $s_\nu\in\R^+$ growing with $\nu$ fast enough to ensure 
\begin{align}
\left\Vert \text{tr}_{d^{\nu\delta m}}\circ\tilde{\Cm}^{\nu}_l - \tilde{\Em}_\nu\circ\tilde{\Dm}_\nu\circ \text{tr}_{d^{\nu\delta m}}\right\Vert_\diamond\leq \frac{1}{\nu}.
\label{equ:ApproxCoding}
\end{align}
With the maps $\Em'_\nu:\M_{2^{R\nu}}\ra \M_{d^{\nu\lb 1+\delta m\rb}}$ and $\Dm'_\nu : \M_{d^{\nu\lb 1+\delta m\rb}}\ra \M_{2^{R\nu}}$ defined via $\Em'_\nu\lb\rho\rb = \tilde{\Em}_\nu\lb\rho\rb\otimes \l|\tilde{\phi}_\nu\rk\lk \tilde{\phi}_\nu\r|$ and $\Dm'_\nu\lb\rho\rb = \tilde{\Dm}_\nu\lb\text{tr}_{d^{\nu\delta m}}\lb\rho\rb\rb$ we get
\begin{align*}
&\inf_{k,\Em,\Dm,\Cm_1,\ldots,\Cm_k} \left\Vert \i^{\otimes n_\nu}_2 - \Dm\circ\prod^k_{l=1} \lb \Cm_l\circ\Tm^{\otimes m_\nu}_{\frac{t}{k}}\rb\circ \Em\right\Vert_{\diamond} \\
&\leq \left\Vert \i^{\otimes n_\nu}_2 - \Dm'_\nu\circ\prod^K_{l=1}\lb \tilde{\Cm}^\nu_l\circ\Tm^{\otimes m_\nu}_{\frac{t}{K}}\rb\circ \Em'_\nu\right\Vert_{\diamond} \\
&= \left\Vert \i^{\otimes n_\nu}_2 - \tilde{\Dm}_\nu\circ\Tm^{\otimes \nu}_{\frac{t}{K}} \circ\text{tr}_{d^{\nu\delta m}}\circ\prod^{K-1}_{l=1}\lb \tilde{\Cm}^\nu_l\circ\Tm^{\otimes m_\nu}_{\frac{t}{K}}\rb\circ \Em'_\nu\right\Vert_{\diamond} \\
&\leq \left\Vert \i^{\otimes n_\nu}_2 - \tilde{\Dm}_\nu\circ\Tm^{\otimes \nu}_{\frac{t}{K}} \circ\prod^{K-1}_{l=1}\lb \tilde{\Em}_\nu\circ\tilde{\Dm}_\nu\circ\Tm^{\otimes \nu}_{\frac{t}{K}}\rb\circ \tilde{\Em}_\nu\right\Vert_{\diamond} + \frac{K}{\nu} \\
&\leq K \left\Vert \i^{\otimes n_\nu}_2 - \tilde{\Dm}_\nu\circ \Tm^{\otimes m_\nu}_{\frac{t}{K}}\circ \tilde{\Em}_\nu\right\Vert_{\diamond} + \frac{K}{\nu}\rightarrow 0 \text{ as } \nu\ra\infty. 
\end{align*}
For the first inequality we just inserted the coding maps constructed before. The second inequality used the triangle inequality K-times and the approximation (\ref{equ:ApproxCoding}). The rest of the proof works the same as for Theorem \ref{thm:CPTPSubCap}. The above calculation shows, that $\tilde{R} = \frac{\log\lb d\rb - \delta}{1+\delta m}$ is an achievable rate for any $\delta >0$. Therefore we have the capacity $\mathcal{Q}_{\sgroup^\ast}\lb t\Lm\rb = \log(d)$, which finishes the proof.

\end{proof}

\section{Quantum subdivision capacity with unitary add-ons}
\label{sec:5}

\subsection{Upper bound via entropy production}
In this section we will consider the subdivision capacity with $\chanSet = \unitary$, i.e., where we are only allowed to use unitary quantum channels to appear in the intermediate coding steps $\Cm^\nu_l$ of (\ref{equ:CodingSubdivisionCap}). Unlike the previous two sections we will show, that the subdivision capacity in this case is in general not equal to the dimension upper bound of $\log(d)$. 

To prove this we consider a Liouvillian $\Lm^{\text{dep}}_r:\M_d\ra\M_d$, which depolarizes onto $\rho_0\in\D_d$ and is given by 
\begin{align}
\Lm^{\text{dep}}_r\lb\rho\rb = r\lb \text{tr}\lb\rho\rb \rho_0 - \rho\rb.
\label{equ:depLiou}
\end{align}
Here $r\in\R^+$ denotes a rate with which the depolarizing noise is applied. The quantum channels generated by such Liouvillians are of the form $\Tm_t\lb\rho\rb = e^{t\Lm^{\text{dep}}_r}\lb\rho\rb = (1-e^{-rt})\text{tr}\lb\rho\rb\rho_0 + e^{-rt}\rho$.

For Liouvillians of the above form we will show the following 
\begin{thm}(Upper bound for depolarizing Liouvillians)
\label{thm:UnitarySubDCapDepolBound}
Let $r\in\R^+$ be given. For the Liouvillian $\Lm^{\text{dep}}_r :\M_{d}\ra\M_d$ defined as in (\ref{equ:depLiou}) we have
\begin{align*}
\mathcal{Q}_\unitary\lb t\Lm^{\text{dep}}_r\rb\leq \log(d) - \lb 1-e^{-rt}\rb S\lb\rho_0\rb.
\end{align*}
\end{thm}

Before proving the theorem we will need a Lemma proved in ~\cite[Lemma 8]{Aharonov1996} for the special case $\rho_0 = \frac{\mathbbm{1}_d}{d}$. The proof of the following slightly more general version proceeds in the same way. 

\begin{lem}[Entropy growth by local depolarizing channels, see ~\cite{Aharonov1996}]
\label{lem:EntropyProduction}
Let $r\in\R^+$ and $\rho_0\in\D_d$ be given. Consider the quantum channel $\Tm_t:\M_d\ra \M_d$ defined as $\Tm_t\lb\rho\rb := e^{t\Lm^{\text{dep}}_r}\lb\rho\rb = \lb 1-e^{-rt}\rb\text{tr}\lb\rho\rb\rho_0 + e^{-rt}\rho$ and any state $\rho\in\D\lb\C^{d^m}\rb$. Then we have
\begin{align*}
S\lb \Tm^{\otimes m}_t\lb\rho\rb\rb &\geq e^{-rt}S\lb\rho\rb + \lb 1-e^{-rt}\rb mS\lb\rho_0\rb\\ 
&\geq \lb 1-e^{-rt}\rb mS\lb\rho_0\rb.
\end{align*}
\end{lem}

This lemma shows the entropy produced by $m$ copies of the quantum channel $\Tm_t = e^{t\Lm^{\text{dep}}_r}$ to be lower bounded by $\lb 1-e^{-rt}\rb mS\lb\rho_0\rb$, which tends to $m S\lb\rho_0\rb$ for $t\ra \infty$. 

In the following we use the fact that unitaries, which form the intermediate coding steps, cannot remove entropy from the system. This entropy growth gives a limit on the maximal possible achievable rate after which a faithful decoding is not possible anymore.

\begin{proof}$\lb\text{of Theorem \ref{thm:UnitarySubDCapDepolBound}} \rb$

For fixed $t\geq 0$ assume that $R>0$ is an achievable rate, in the sense of Definition \ref{defn:SubdivisionCap}, for the depolarizing Liouvillian of rate $r\in\R^+$ onto the state $\rho_0\in\D_d$ as in (\ref{equ:depLiou}). Note that by an argument in ~\cite{Kretschmann2004}, which works the same way for the quantum subdivision capacities, it is enough to consider the sequences $n_\nu = R\nu$ and $m_\nu = \nu$ in (\ref{equ:CodingSubdivisionCap}) when testing whether a certain rate $R<\mathcal{Q}_{\unitary}\lb t\Lm\rb$ is achievable. Thus for the sequences  $n_\nu = R\nu$, $m_\nu = \nu$ and quantum channels $\Tm_{\frac{t}{k}}:=e^{\frac{t}{k}\Lm^{\text{dep}}_r}$ we have  
\begin{align}
\inf_{k,\Em,\Dm,\Um_1,\ldots,\Um_k} \left\Vert \i^{\otimes R\nu}_d - \Dm\circ \prod^k_{l=1} \lb\Um_l\circ \Tm_{\frac{t}{k}}^{\otimes \nu}\rb\circ \Em\right\Vert_{\diamond} \leq \epsilon_\nu.
\label{equ:ErrorCoding}
\end{align}
For any fixed coding scheme of length $K\in\N$ achieving the above bound and which is defined by quantum channels $\Em_\nu,\Dm_\nu$ and unitary channels $\Um^\nu_1,\ldots,\Um^\nu_K\in\unitary$, consider the quantum channels $\tilde{\Tm}_{K}^\nu := \prod^{K}_{l=1} \lb\Um^\nu_l\circ \Tm_{\frac{t}{K}}^{\otimes \nu}\rb\circ \Em_\nu$. These map the input state to the state right before the final decoding operation. Applying Lemma \ref{lem:EntropyProduction} for the depolarizing channel $K$ times and using that the entropy is invariant under unitary transformations, we obtain
\begin{align}
S\lb\tilde{\Tm}_{K}^\nu\lb\rho\rb\rb \geq \lb 1-e^{-rt}\rb \nu S\lb\rho_0\rb 
\label{equ:entropyBound}
\end{align} 
for any quantum state $\rho\in \M^{\otimes R\nu}_{2}$. 

Now consider the maximally mixed input state $\sigma = \sum^{2^{R\nu}}_{i=1} \frac{1}{2^{R\nu}}\l|i \rk\lk i\r|$ for some orthonormal basis $\lset\l| i\rk\rset^{2^{R\nu}}_{i=1}$ of $\C^{2^{R\nu}}$. Inserting this state we obtain
\begin{align*}
\nu\log(d)&\geq S\lb\tilde{\Tm}_{K}^\nu\lb \sigma\rb\rb  = \chi\lb \lset\frac{1}{2^{R\nu}}, \tilde{\Tm}_{K}^\nu\lb\l| i\rk\lk i\r|\rb\rset\rb + \sum^{2^{R\nu}}_{i=1} \frac{S\lb\tilde{\Tm}_{K}^\nu\lb \l| i\rk\lk i\r|\rb\rb}{2^{R\nu}} 
\end{align*} 
where the first inequality is the dimension upper bound on the entropy of a $d^\nu$-dimensional system and where we introduced Holevo's $\chi$-quantity~\cite[p. 531]{Nielsen2007}. This is defined for an ensemble of density operators $\lset p_i,\rho_i\rset^N_{i=1}$ as 
\begin{align*}
\chi\lb\lset p_i,\rho_i\rset\rb := S\lb\sum^N_{i=1} p_i\rho_i\rb - \sum^N_{i=1} p_i S\lb\rho_i\rb.
\end{align*}
Using the data processing inequality for the $\chi$-quantity, see ~\cite[p. 602]{Nielsen2007}, and the entropy bound (\ref{equ:entropyBound}), we obtain 
\begin{align}
\nu\log(d) &\geq \chi\lb \lset\frac{1}{2^{R\nu}}, \Dm_\nu\circ\tilde{\Tm}_{K}^\nu\lb\l| i\rk\lk i\r|\rb\rset\rb + \lb 1-e^{-rt}\rb \nu S\lb\rho_0\rb \nonumber\\
&= S\lb \Dm_\nu\circ\tilde{\Tm}_{K}^\nu\lb\sigma\rb \rb - \sum^{2^{R\nu}}_{i=1} \frac{1}{2^{R\nu}} S\lb\Dm_\nu\circ\tilde{\Tm}_{K}^\nu\lb \l| i\rk\lk i\r|\rb\rb + \lb 1-e^{-rt}\rb \nu S\lb\rho_0\rb 
\label{equ:Estimate}
\end{align}

Finally we can apply the Fannes-Audenaert inequality~\cite{Audenaert2007}, i.e. for all quantum states $\rho_1,\rho_2\in\D_d$ with trace distance $\delta = \frac{1}{2}\left\Vert\rho_1 -\rho_2\right\Vert_1$ we have
\begin{align}
| S\lb\rho_1\rb - S\lb\rho_2\rb | \leq \delta\log(d) + H\lb \delta\rb 
\label{equ:FAInequ}
\end{align}
for the binary entropy $H\lb\delta\rb = -\delta\log(\delta) - (1-\delta)\log(\delta)$. Estimating the entropies in (\ref{equ:Estimate}) using (\ref{equ:FAInequ}) and (\ref{equ:ErrorCoding}) leads to
\begin{align*}
S\lb\Dm_\nu\circ\tilde{\Tm}_{K}^\nu\lb \l| i\rk\lk i\r|\rb\rb \leq \frac{1}{2}\epsilon_\nu + 1
\end{align*}
for all $i\in\lset 1,\ldots,2^{R\nu}\rset$ and 
\begin{align*}
S\lb \Dm_\nu\circ\tilde{\Tm}_{K}^\nu\lb\sigma\rb \rb \geq R\nu - \frac{1}{2}\epsilon_\nu - 1.
\end{align*}
Inserting these bounds in (\ref{equ:Estimate}) we obtain 
\begin{align}
\nu\log(d) \geq R\nu + \lb 1-e^{-rt}\rb \nu S\lb\rho_0\rb - \epsilon_\nu - 2
\label{equ:FinalEstimate} 
\end{align} 
where the approximation error fulfills $\epsilon_\nu \ra 0$ as $\nu\ra \infty$ if the coding scheme achieves the rate $R$. Dividing through $\nu$ in (\ref{equ:FinalEstimate}) and taking the limit $\nu\ra\infty$ leads to the bound $R < \log(d) - (1-e^{-rt})S\lb\rho_0\rb$, which finishes the proof.

\end{proof}

Note that the above proof also shows that it is not even possible to transmit classical information with a rate higher than the bound given in Theorem \ref{thm:UnitarySubDCapDepolBound} through the subdivision coding scheme. It is possible to define subdivision capacities for the transmission of classical information through a  quantum dynamical semigroup in a similar way as the quantum subdivision capacities in Definition \ref{defn:SubdivisionCap}. The corresponding capacity for unitary intermediate coding operations would then be an upper bound on $\mathcal{Q}_\unitary$ in general. For a depolarizing Liouvillian onto $\rho_0\in\D_d$, see (\ref{equ:depLiou}), the above proof shows that the classical subdivision capacity is smaller than $\log(d) - (1-e^{-rt})S\lb\rho_0\rb$.  

Applying Theorem \ref{thm:UnitarySubDCapDepolBound} to the completely depolarizing Liouvillian, i.e. the depolarizing Liouvillian onto a maximally mixed state $\rho_0=\frac{\mathbbm{1}_d}{d}\in\D_d$, see (\ref{equ:depLiou}), gives the following corollary.

\begin{cor}[$\mathcal{Q}_\unitary$ arbitrarily small]
Let $r\in\R^+$ be given. For the completely depolarizing Liouvillian $\Lm^{\text{cd}}_r :\M_{d}\ra\M_d$ defined as $\Lm^{\text{cd}}_r\lb\rho\rb := r\lb \text{tr}\lb\rho\rb\frac{\mathbbm{1}_d}{d} - \rho\rb$  we have
\begin{align*}
\mathcal{Q}_\unitary\lb t\Lm^{\text{cd}}_r\rb\leq e^{-rt} \log(d)\,.
\end{align*}
\end{cor}
This corollary shows that, unlike $\mathcal{Q}_\chan$ and $\mathcal{Q}_{\sgroup^\ast}$, the subdivision capacity $\mathcal{Q}_\unitary$ can become arbitrarily small. In the following subsection we will prove lower bounds showing that $\mathcal{Q}_\unitary$ is always strictly larger than zero. 

\subsection{Lower bound via finite subdivisions}
Here we will prove lower bounds on $\mathcal{Q}_\unitary(t\Lm)$ by exploiting the fixed points of the noise Liouvillian $\Lm:\M_d\ra\M_d$, i.e., points $\rho_0\in\D_d$ fulfilling $\Lm(\rho_0)=0$. These lower bounds show that $\mathcal{Q}_\unitary\lb t\Lm\rb>0$ for all $t\in\R^+$ and any noise Liouvillian $\Lm$. Note that almost all noise Liouvillians $\Lm$ lead to time-evolutions $e^{t\Lm}$ which become entanglement breaking eventually after a finite time. This implies that $\mathcal{Q}(e^{t\Lm})=0$ for all $t\geq t_0$. Therefore our lower bounds prove that for almost all $\Lm$ there is a finite time $t$ such that $\mathcal{Q}_\unitary\lb t\Lm\rb>\mathcal{Q}\lb e^{t\Lm}\rb = 0$ holds.

We will start by considering the case of a pure fixed point $\l|\psi\rk\lk \psi\r|\in\D_d$ of the noise Liouvillian $\Lm:\M_d\ra\M_d$. For any number $k\in\N$ of subdivisions in the sense of Definition \ref{defn:SubdivisionCap}, we try to implement the coding scheme used for the proof of Theorem \ref{thm:SemigroupSubCap}. This coding scheme implements the successive decoding and encoding operations used for the proof of Theorem \ref{thm:CPTPSubCap} via unitaries and pure ancillas. Within this coding scheme, these pure ancillas are created using certain depolarizing channels, but as these channels are not unitary, we cannot use them in a coding scheme for $\mathcal{Q}_\unitary\lb t\Lm\rb$. Nevertheless we can implement the coding scheme if we simply prepare sufficiently many copies of the pure fixed point $\l|\psi\rk\lk \psi\r|\in\D_d$ of the noise Liouvillian $\Lm$ in the encoding stage $\Em$ at the beginning of the coding scheme. As fixed points of the noisy time-evolution they are not disturbed, and they can be used as ancillas in the $k$ subdivision steps, see Fig. \ref{fig:UnitaryCodingSchemeFix}. This idea leads to our lower bound on $\mathcal{Q}_\unitary\lb t\Lm\rb$. This may be further improved by ``recycling'' the used ancilla states, but for simplicity we will not pursue this here.

\begin{figure*}
\center
     \includegraphics[width=0.8\textwidth]{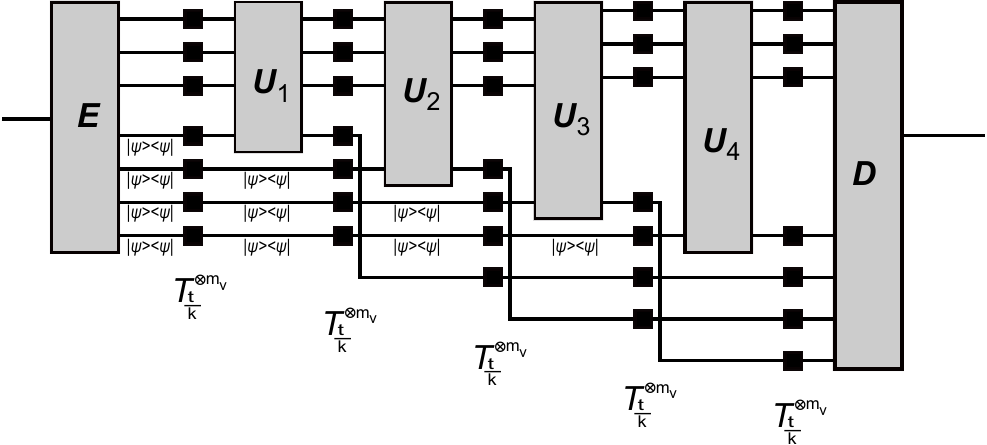}
     \caption{Unitary Coding Scheme for $k=4$, exploiting the pure fixed point $\l|\psi\rk\lk \psi\r|$ of the time-evolution. As $\l|\psi\rk\lk \psi\r|$ is not disturbed by the noise we can introduce enough copies of this state in the beginning to implement the coding scheme from Theorem \ref{thm:CPTPSubCap} with unitaries and the pure states $\l|\psi\rk\lk \psi\r|$. }
     \label{fig:UnitaryCodingSchemeFix}
\end{figure*}

Note that the number of pure ancillas needed to implement the $k$ intermediate coding operations in the proof of Theorem \ref{thm:SemigroupSubCap} depends on the difference 
\begin{align}
\delta_k(t,\Lm) := \log(d)- I^{\text{coh}}\lb \omega^{A'A} ,\i_{A'}\otimes \Tm_{t/k}\rb
\end{align}
for the quantum channel $\Tm_{t/k}:=e^{\frac{t}{k}\Lm}$, with the coherent information from Definition \ref{defn:Entropies} and where $k\in\N$ denotes the number of subdivisions, see Definition \ref{defn:ContQuantCap}. 

For any fixed number $k\in\N$ of subdivisions the proof of Theorem \ref{thm:SemigroupSubCap} gives a coding scheme using pure ancilla qubits at a rate of $\delta_k(t,\Lm)(1+c)$  in each of the $k$ coding steps, where $c\in\R$ denotes the constant introduced in the typical subspace argument preceding the proof of Theorem \ref{thm:SemigroupSubCap}, see Lemma \ref{lem:ChoiMatrSmallEnv}. In the limit of asymptotically many parallel channel uses, the achievable rate $I^{\text{coh}}\lb \omega^{A'A} ,\i_{A'}\otimes \Tm_{t/k}\rb$ from the proof of Theorem \ref{thm:SemigroupSubCap}, goes thus down by a factor of $\lb 1 + k\frac{\delta_k(t,\Lm)(1+c)}{\log(d)}\rb$ because for each faithfully transmitted qubit one needs $k\frac{\delta_k(t,\Lm)(1+c)}{\log(d)}$ local noise channels to transmit the pure ancilla states. This gives 

\begin{thm}[Lower bound with pure fixed points]
\label{thm:lowerBound1}
Let $\Lm:\M_d\ra\M_d$ denote a Liouvillian with a pure fixed point $\l|\psi\rk\lk\psi\r|\in\D_d$. Then we have 
\begin{align*}
\mathcal{Q}_\unitary\lb t\Lm \rb \geq \sup_k \lb\frac{I^{\text{coh}}\lb \omega^{A'A} ,\i_{A'}\otimes \Tm_{t/k}\rb\log(d)}{\log(d)+k\delta_k(t,\Lm)(1+c)}\rb.
\end{align*} 
\end{thm} 

Note that by continuity 
\begin{align*}
I^{\text{coh}}\lb \omega^{A'A} ,\i_{A'}\otimes \Tm_{t/k}\rb\ra \log(d)>0
\end{align*}
as $k\ra \infty$. Therefore we have $\mathcal{Q}_\unitary\lb t\Lm \rb>0$ for any Liouvillian $\Lm$ with a pure fixed point and any $t\in\R^+$.

To generalize the above scheme to noise Liouvillians $\Lm:\M_d\ra\M_d$ with arbitrary fixed points $\rho_0\in\D_d$ fulfilling $S(\rho_0)<\log(d)$ we will use Schumacher compression~\cite{Schumacher1995}, see also~\cite[Theorem 12.6]{Nielsen2007}. Given many independent copies $\rho^\nu_0$ of a quantum state $\rho_0\in\D_d$, the theory of typical subspaces~\cite{Devetak2005,Bennett2012,Winter1999} can be applied. For any $\delta>0$ there exists a typical projector $\Pi_{\delta}$ onto a $\delta$-typical subspace $\mathfrak{T}_\delta$ with respect to $\rho_0$, i.e.:

\begin{enumerate}
\item We have $\text{tr}\lb\Pi_\delta \rho_0^{\otimes \nu}\rb\geq 1-2^{-\nu c'\delta^2}$ for some constant $c'\in\R$ independent of $\nu$. 
\item Furthermore $\text{tr}\lb\Pi_\delta\rb\leq 2^{\nu\lb S(\rho_0)+c\delta\rb}$ for some constant $c\in\R$ independent of $\nu$.
\end{enumerate}

These properties imply that we can write
\begin{align*}
\rho_0^{\otimes \nu} = p_\nu \rho_{\text{typ}} + (1-p_\nu)\sigma
\end{align*} 
for $p_\nu = \text{tr}\lb\Pi_\delta \rho_0^{\otimes \nu}\rb$, a state $\rho_{\text{typ}}$ supported on the typical subspace and some state $\sigma\in\D\lb\mathfrak{T}^\perp_\delta\rb$ not supported on the typical subspace. The property 1.\ from above implies that $p_\nu\ra 1$ exponentially as $\nu\ra \infty$. By 2.\ from above we can choose a unitary $U$ mapping the typical subspace onto vectors of the form $\l|\cdot\rk\otimes \l| 0\rk$ where the first tensor factor contains the compressed information on $\nu\lb S(\rho_0)+c\delta\rb$ qubits and the second tensor factor is in a pure state. This yields
\begin{align*}
U\rho_0^{\otimes\nu}U^\dagger =p_\nu \rho_{\text{compr}}\otimes \l| 0\rk\lk 0\r| + (1-p_\nu)U\sigma U^\dagger
\end{align*} 
Therefore Schumacher compression generates a pure state $\l| 0\rk\lk 0\r|$ on the space of $\nu\lb \log(d) - S\lb\rho_0\rb - c\delta\rb$ qubits with fidelity exponentially good in $\nu$. We will use these states as almost pure ancillas. The idea of the following coding scheme is to prepare sufficiently many copies of the fixed point $\rho_0\in\D_d$ of the noise Liouvillian $\Lm$ in the beginning of the scheme. As fixed points they are not disturbed by the noisy time-evolution and their entropy does not increase with time. We can apply Schumacher compression on many copies of $\rho_0$ to generate almost pure ancillas right before they are needed. The maximal rate of generating these almost pure ancilla qubits can be derived from the dimension of the typical subspace introduced above. As $\delta>0$ for the Schumacher compression protocol can be chosen arbitrarily small, the rate of generating almost pure ancilla qubits is $\log(d)-S(\rho_0)$ per local noise channel in the fixed point $\rho_0$. Note that we cannot create any pure ancillas from a maximally mixed fixed point $\rho_0 = \frac{\mathbbm{1}_d}{d}$ and this coding scheme does not work in that case.

As in the case of pure fixed points we can compute the rates of the above coding scheme. This gives: 

\begin{thm}[Lower bound with fixed points]
\label{thm:lowerBound2}
Let $\Lm:\M_d\ra\M_d$ denote a Liouvillian with fixed point $\rho_0\in\D_d$. Then we have 
\begin{align*}
&\mathcal{Q}_\unitary\lb t\Lm \rb \geq \sup_k \lb\frac{I^{\text{coh}}\lb \omega^{A'A} ,\i_{A'}\otimes \Tm_{t/k}\rb\lb \log(d)-S(\rho_0)\rb}{\log(d)-S(\rho_0)+k\delta_k(t,\Lm)(1+c)}\rb.
\end{align*} 

\end{thm}  

Note that Theorem \ref{thm:lowerBound1} is a special case of Theorem \ref{thm:lowerBound2} for the case $S(\rho_0)=0$. As before, this proves that $\mathcal{Q}_\unitary\lb t\Lm\rb >0$ for any Liouvillian $\Lm$ with a fixed point that is not maximally mixed and any $t\in\R^+$. 

The coding schemes introduced here may be improved further by reusing the ancillas again. When the noisy time-evolution generated by the noise Liouvillian $\Lm$ has a limit point of not maximal entropy, i.e., there exists a state $\rho_\infty$ fulfilling $S(\rho_\infty)<\log(d)$ such that $e^{t\Lm}\lb\rho\rb\ra \rho_\infty$ as $t\ra\infty$, this can be used to ``cool'' the system in the sense of \cite{BenOr2013}. We could for instance let the used ancillas be affected by the noise, which drives them towards $\rho_\infty$. If they are close enough to $\rho_\infty$ we can again use Schumacher compression and re-use them to obtain some pure ancillas to be used in the coding scheme. 

The above theorem proves $\mathcal{Q}_\unitary\lb t\Lm\rb>0$ only when the fixed point $\rho_0\in\D_d$ of the noise Liouvillian $\Lm$ to be not maximally mixed, i.e., $S(\rho_0)<\log(d)$. To obtain a similar coding scheme yielding a non-zero rate for a maximally mixed fixed point $\rho_0 = \frac{\mathbbm{1}_d}{d}$ we could again introduce pure ancillas in the encoding stage $\Em$ of the coding scheme. As the noise acts on these ancillas their entropy will grow but never reaches the maximum $\log(d)$ in finite time. As the entropy grows the number of pure ancillas that we may obtain via Schumacher compression decreases. Still when enough ancillas are introduced in the beginning of the protocol, we can use Schumacher compression to obtain all the pure ancillas we need. Therefore a coding scheme as before can be implemented showing that $\mathcal{Q}_\unitary(t\Lm)>0$ even in this case.

\section{Continuous-time coding schemes}
\label{sec:6}

In the previous sections we studied the quantum subdivision capacities $\mathcal{Q}_\chanSet$, Definition \ref{defn:SubdivisionCap}, with respect to different subsets $\chanSet\subseteq \chan$ of quantum channels, from which the intermediate coding channels $\Cm_l$  were taken, see (\ref{equ:CodingSubdivisionCap}). These capacities are all discrete in the sense that the infimum in Definition \ref{defn:ContQuantCap} captures only countably many subdivisions. In this section we present a variation of the quantum subdivision capacities to account for ``time-continuous'' coding schemes. For a rough idea, consider quantum systems undergoing a noisy time-evolution generated by some Liouvillian $\Lm:\M_d\ra\M_d$. To quantify the optimal rate of information transmission using this quantum system and ``time-continuous'' coding we consider the limit of infinitely many copies of the system and allow for additional ``correction'' generated by Liouvillians $\Lm_c:\M_{d^m}\ra\M_{d^m}$ acting on all $m$ copies of the system. We will show, the time-evolution generated by the noise Liouvillian $\Lm$ together with that \emph{coding Liouvillian} $\Lm_c$ acting on many copies of the system can improve the transmission of quantum information compared to the time-evolution generated by the noise $\Lm$ alone.     

Similar schemes to use supplementary Liouvillians to improve the storage time of quantum memories have been studied before ~\cite{Pastawski2009,Pastawski2011}, albeit not in the asymptotic setting. 

We will state our initial definition in the most general form for time-dependent Liouvillians and this definition will include the subdivision capacities from  Definition \ref{defn:SubdivisionCap} when applied to semigroups as in Sec.\ref{sec:4}. Later we will provide an example of a time-\emph{in}dependent Liouvillian and a scheme, which can improve the transmission of quantum information using only a time-\emph{in}dependent coding Liouvillian.

\subsection{Continuous quantum capacity}
\label{sec:6.1}
Consider a time-dependent Liouvillian~\cite{Rivas2012}, i.e. a map $\Lm:\R^+\times\M_d\ra\M_d$ of the form 
\begin{align}
\Lm\lb t,\rho\rb =& -i\lbr H(t),\rho\rbr + \sum^N_{k=1}\bigg{(} A_k(t)\rho A_k(t)^\dagger -\frac{1}{2}A_k(t)^\dagger A_k(t)\rho - \frac{1}{2}\rho A_k(t)^\dagger A_k(t)\bigg{)}
\label{equ:TDLiouvillian}
\end{align}
where $H:\R^+\ra \mathcal{H}_d$ is a piecewise continuous map into the set of Hermitian matrices and $A_k:\R^+\ra \M_d$ are piecewise continuous maps into $d\times d-$matrices. In the following, for any $t\geq 0$, we will write $\Lm(t):\M_d\ra\M_d$ to denote the map $\Lm(t,\cdot)$ as defined above. We will denote the set of all time-dependent Liouvillians (\ref{equ:TDLiouvillian}) as $\mathfrak{TDL}$. Note that the case (\ref{equ:Liouvillian}) from before corresponds to $\Lm(t)\equiv\Lm$.       

For a quantum state $\rho_0\in\D_d$ the solution of
\begin{align}
\frac{d}{dt}\rho(t) = \Lm(t)\rho(t)
\label{equ:Master}
\end{align} 
with the initial condition $\rho(0)=\rho_0$ describes the time-evolution of the initial state $\rho_0$ under the time-dependent Liouvillian $\Lm$ of the form (\ref{equ:TDLiouvillian}), see ~\cite{Rivas2012}. The equation (\ref{equ:Master}) is called the master equation of the evolution generated by the Liouvillian $\Lm$. For each $t\geq 0$ we may define a quantum channel $\Tm_t :\M_d\ra \M_d$ mapping each initial state $\rho_0$ to its time-evolved state $\rho(t)=\Tm_t\lb\rho_0\rb$. This quantum channel is $\Tm_t\lb\rho\rb = T\exp\lb \int^t_0 \Lm\lb t'\rb dt'\rb\lb\rho\rb$, where $T\text{exp}$ denotes the time-ordered exponential function~\cite{Rivas2012}. 

To define a capacity for transmitting quantum information over a system affected by continuous noise $\Lm$, we consider the limit of many copies of the system, each of them affected by the noise independently, as in Definition \ref{defn:QuantCap}. It can be easily seen that the global noise resulting from taking a tensor power of $m$ systems of dimension $d$ each affected by the local noise Liouvillian $\Lm:\M_d\ra\M_d$ is again generated by a time-dependent noise Liouvillian denoted by $\Lm^{\Ltimes m} :\R^+\times \M_{d^m}\ra\M_{d^m}$. This Liouvillian can be expressed in terms of the 'local' noise Liouvillian $\Lm$ via 
\begin{align}
\Lm^{\Ltimes m}\lb t\rb := \sum^m_{i=1} \i_{m\setminus i} \otimes \Lm\lb t\rb 
\label{equ:TensorLiou}
\end{align}
where $\i_{m\setminus i} \otimes \Lm\lb t\rb: \M^{\otimes m}_d\ra \M^{\otimes m}_d$ denotes the operator acting as the identity on all tensor factors except the one indexed by $i\in\lset 1,\ldots,m\rset$, where $\Lm(t)$ is applied.     

As for the subdivision capacities we will state our definition in the most general framework by considering arbitrary subsets $\chanSet\subseteq \mathfrak{TDL}$ of time-dependent coding Liouvillians. 

\begin{defn}[Continuous quantum capacity $\mathcal{Q}_\chanSet^\text{cont}$]
\label{defn:ContQuantCap}

For $\chanSet\subseteq \mathfrak{TDL}$ the $\chanSet-$\textbf{continuous quantum capacity} of a time-dependent Liouvillian $\Lm: \R^+\times\M_d\ra \M_d$, see (\ref{equ:TDLiouvillian}), is defined as 
\begin{align*}
\mathcal{Q}^\text{cont}_\chanSet\lb t\Lm\rb := \sup\lset R\in \R :\text{ R achievable rate}\rset
\end{align*}
where a rate $R\in\R^+$ is called achievable if there exist sequences $\lb n_\nu\rb^\infty_{\nu=1},\lb m_\nu\rb^\infty_{\nu=1}$ such that $R = \limsup_{\nu\ra \infty}\frac{n_\nu}{m_\nu}$ and 
\begin{align*}
\inf &\left\Vert \i^{\otimes n_\nu}_2 - \Dm\circ T\exp\lb \int^t_0 \Lm^{\Ltimes m_\nu}\lb t'\rb + \Lm_c\lb t'\rb  dt'\rb \circ \Em\right\Vert_{\diamond} \ra 0 \text{ as } \nu\ra \infty.
\end{align*}
Here the infimum is over arbitrary encoding and decoding quantum channels $\Em:\M^{\otimes n_\nu}_{2}\ra \M^{\otimes m_\nu}_{d}$ and $\Dm:\M^{\otimes m_\nu}_{d}\ra \M^{\otimes n_\nu}_{2}$ and appropriate time-dependent coding Liouvillians $\Lm_c\in \chanSet$ from the chosen subset.
\end{defn} 

Using the ideas from the proof of Theorem \ref{thm:SemigroupSubCap}, i.e., that the $\sgroup^\ast$-quantum subdivision capacity on a $d$-dimensional system is $\log(d)$, we can also show that $\mathcal{Q}^\text{cont}_\mathfrak{TDL}\lb t\Lm\rb = \log(d)$ for any $t\in\R^+$ and any time-dependent noise Liouvillian $\Lm:\M_d\ra\M_d$ of the form (\ref{equ:TDLiouvillian}). To see this consider the time-evolution according to the Liouvillian $\Lm$ which can be subdivided into quantum channels of the form  
\begin{align*}
\Tm_{t_1,t_2} = T\exp\lb \int^{t_2}_{t_1} \Lm\lb t'\rb dt'\rb.
\end{align*}

For any $\delta>0$ there exists, by continuity, a $K\in\N$ such that every quantum channel $\Tm_{\frac{t(l-1)}{K},\frac{tl}{K}}$ of the time-evolution fulfills $I^{\text{coh}}\lb \omega^{A'A} ,\i_{A'}\otimes \Tm_{\frac{t(l-1)}{K},\frac{tl}{K}}\rb> \log(d)-\delta$. Thus we can construct a coding scheme as in the proof of Theorem \ref{thm:SemigroupSubCap} achieving the rate $R=\log(d)-\delta$ and which can be written as
\begin{align*}
\Dm_\nu\circ \prod^K_{l=1} \lb\Cm^\nu_l\circ \Tm^{\otimes m_\nu}_{\frac{t(l-1)}{K},\frac{tl}{K}}\rb\circ \Em_\nu
\end{align*} 
where $m_\nu$ is a sequence denoting the number of channel uses and $\Dm_\nu,\Em_\nu,\Cm^\nu_l$ are coding channels in the sense of Definition \ref{defn:SubdivisionCap}. As  $\Cm^\nu_l\in\sgroup^\ast$ there are time-dependent Liouvillians $\Lm^\nu_l$ such that $\Cm^\nu_l = T\exp\lb \int^s_0 \Lm^\nu_l\lb t'\rb dt'\rb$ for a time $s\in\R^+$ that can be chosen arbitrarily small when only the strength $\|\Lm^\nu_l\|$ is chosen high enough. Now we are finished by approximating
\begin{align*}
\prod^K_{l=1} \lb\Cm^\nu_l\circ \Tm^{\otimes m_\nu}_{\frac{t(l-1)}{K},\frac{tl}{K}}\rb \simeq T\exp\lb \int^t_0 \Lm^{\Ltimes m_\nu}\lb t'\rb + \Lm_c\lb t'\rb  dt'\rb
\end{align*}
with a time-dependent coding Liouvillian $\Lm_c$ that turns on the Liouvillians $\Lm^\nu_l$ implementing the coding maps $\Cm^\nu_l$ from above at the right times. To get a vanishing error in Definition \ref{defn:ContQuantCap} the strength $\|\Lm_c\|$ can be chosen arbitrarily high compared to that of $\Lm$.

There are some obvious operational restrictions one might impose on the set $\chanSet$ of allowed coding Liouvillians. The first possibility is to put a bound on the strength of the coding Liouvillians in Definition \ref{defn:ContQuantCap}, e.g. $\left\Vert\Lm(t)\right\Vert\leq c$ for all $t\in\R^+$ with a constant $c>0$. Then it would not be possible to generate arbitrarily pure ancilla states on arbitrarily small time-scales, as in the proof of Theorem \ref{thm:SemigroupSubCap}. A second possibility is to restrict to time-\emph{in}dependent coding Liouvillians $\Lm_c(t)\equiv\Lm_c$ in Definition \ref{defn:ContQuantCap}, so that it is not obvious anymore how to implement a coding scheme as in the proof of Theorem \ref{thm:SemigroupSubCap}. A third possibility would be to impose locality constrains on the coding Liouvillians $\Lm_c$ in Definition \ref{defn:ContQuantCap}. Each of the above possibilities would lead to a continuous quantum capacity and it is a priori not clear how they behave and when they coincide. We will now focus on a combination of the last two cases and consider an example of a coding scheme with a time-\emph{in}dependent and local coding Liouvillian $\Lm_c$, specializing Definition \ref{defn:ContQuantCap}, for a time-\emph{in}dependent noise Liouvillian $\Lm$ of a special form. 

\subsection{Coding schemes with quasi-local and time-independent coding Liouvillians}
\label{sec:6.2}
In this section we will discuss how to construct continuous-time coding schemes from quantum error correcting codes for a particular class of noisy time-evolutions. Consider a $d$-dimensional physical system affected by noise given by a quantum channel $\Tm:\M_d\ra\M_d$. If the noise channel acts on the system continuously in time it corresponds to a continuous time-evolution generated by a time-\emph{in}dependent local noise Liouvillian $\Lm:\M_d\ra\M_d$ of the form
\begin{align}
\Lm\lb \rho\rb = \Tm\lb\rho\rb -\rho .
\label{equ:SpecialLiou}
\end{align}  

An $(n,m)$-quantum error correcting code~\cite[pp. 435]{Nielsen2007} for the noise $\Tm$ is formally defined as a $2^n$-dimensional subspace(``codespace'') of the $d^m$-dimensional Hilbert space corresponding to $m$-copies of a given system. We will denote by $\Vm:\M_{2^n}\ra\M_{d^m}$ an isometric embedding, which encodes $n$ 'logical' qubits into the codespace on $m$-copies of the system. Note that the overall noise affecting $m$ copies of the physical system is generated by the Liouvillian

\begin{align}
\Lm^{\Ltimes m} = \Tm^{\Ltimes m} - m \i_{2^m}\,.
\label{equ:LiovillianDep5}
\end{align}

When the noise in each system is generated by the Liouvillian $\Lm$ from (\ref{equ:SpecialLiou}). When the linear map $\Tm^{\Ltimes m}$ acts on the codespace it disturbs the encoded states, taking them outside the codespace. We denote by $\Rm:\M_{d^m}\ra\M_{d^m}$ a recovery quantum channel fulfilling 

\begin{align}
\Rm\circ \Vm = \Vm
\label{equ:Recovery2}
\end{align}
which corrects these errors by mapping the disturbed states back to the codespace, i.e., satisfies

\begin{align}
\Rm\circ \Tm^{\Ltimes m}\circ\Vm = m\Vm.
\label{equ:ChannelErrCorr2}
\end{align}

By assuming the existence of an $(n,m)$-quantum error correcting code we are implicitly making assumptions about the noise channel $\Tm$. However, there are relevant cases of noise channels, where such a code exists. An important instance are channels with unitary Kraus operators such that there is a stabilizer code correcting these unitary errors. For these codes the recovery operation $\Rm$ is given by a projective measurement of the error syndrome and a unitary error correction conditioned on the measured error syndrome. An important example of a channel with unitary Kraus operators is the depolarizing channel $\Tm_{\text{dep}}:\M_{2}\ra \M_{2}$  defined as 
\begin{align}
\Tm_{\text{dep}}\lb\rho\rb = \frac{1}{3}\lb X\rho X+Y\rho Y+Z\rho Z\rb 
\label{equ:DepChannel}
\end{align}
for the Pauli matrices 
\begin{align*}
X = \begin{pmatrix}
0 & 1 \\ 1 & 0
\end{pmatrix}
\hspace*{0.3cm}
Y = \begin{pmatrix}
0 & -i \\ i & 0
\end{pmatrix}
\hspace*{0.3cm}
Z = \begin{pmatrix}
1 & 0 \\ 0 & -1
\end{pmatrix}.
\end{align*}
The error introduced by this channel can be corrected using the 5-qubit stabilizer code~\cite[pp. 468]{Nielsen2007} and we present more details in the example at the end of this section.
 
We will now introduce a continuous coding scheme in the sense of Definition \ref{defn:ContQuantCap} by implementing the recovery operation $\Rm$ of the quantum error correcting code continuously in time. The conditions (\ref{equ:Recovery2}) and (\ref{equ:ChannelErrCorr2}) guarantee that the recovery operation $\Rm$ corrects any error introduced by $\Tm^{\Ltimes m}$ with high fidelity. The time-evolution generated by the sum of the Liouvillian $\Lm$ corresponding to $\Tm^{\Ltimes m}$ and a coding Liouvillian, in the sense of Definition \ref{defn:ContQuantCap}, implementing the recovery operation $\Rm$ can be thought of as applying $\Tm^{\Ltimes m}$ or $\Rm$ very fast after another. By the assumptions made it seems reasonable that by making the rate of the recovery operation high enough any error introduced by the noise can be corrected. 
\newpage 
To make this intuition more precise we prove the following

\begin{thm}[Continuous-time quantum error correction]
\label{thm:ExplicitCodingLiou}

Let $\Tm:\M_d\ra\M_d$ be a quantum channel, such that there exists an $(n,m)$-quantum error correcting code with an isometric embedding $\Vm:\M_{2^n}\ra\M_{d^m}$ and a recovery quantum channel $\Rm:\M_{d^m}\ra\M_{d^m}$ fulfilling the conditions (\ref{equ:Recovery2}) and (\ref{equ:ChannelErrCorr2}). Furthermore let $\chanSet\subseteq\mathfrak{TDL}$ be a set of coding Liouvillians, see Definition \ref{defn:ContQuantCap}, containing $\lb r\lb \Rm -\i_{d^m}\rb\rb^{\Ltimes k}$ for all $k\in\N$ and $r\in\R^+$. 

Then for the noise Liouvillian $\Lm:\M_d\ra\M_d$ of the form $\Lm = \Tm -\i$ and any $t\in\R^+$ we have 
\begin{align}
\mathcal{Q}^\text{cont}_\chanSet\lb t\Lm\rb \geq \frac{n}{m}.
\end{align} 
\end{thm}

\begin{proof}

Note that in the sense of Definition \ref{defn:ContQuantCap} we use the encoding map $\Em=\Vm$. Because $\Vm$ is an isometric embedding, a decoding map $\Dm$ can be chosen as the projection mapping the encoded qubits back to the corresponding logical qubit and mapping states not in the code to an error state. 

The noise acting on the codespace is given by $\Lm^{\Ltimes m}$, see (\ref{equ:LiovillianDep5}). To define a continuous coding scheme as in Definition \ref{defn:ContQuantCap} we consider the coding Liouvillian $\Lm_c:\M_{d^m}\ra\M_{d^m}$ as
\begin{align*}
\Lm_c := r\lb\Rm - \i_{d^m}\rb.
\end{align*}
for a rate $r\in\R^+$ controlling the strength of the error correction. As in Definition \ref{defn:ContQuantCap} we consider the Liouvillian $\Lm^{\Ltimes m}$ supplemented by the coding Liouvillian $\Lm_c$. We obtain
\begin{align*}
&\Rm\circ e^{t\lb \Lm^{\Ltimes m} + r\Lm_c\rb}\circ\Vm = e^{-t(r+m)} \Rm\circ\sum^\infty_{k=0} \frac{t^k}{k!}\lb\Tm^{\Ltimes m} + r\Rm\rb^k\circ \mathcal{V}. 
\end{align*}
It is clear that we can write this expression as a convex combination of $\Vm$ and a suitable quantum channel $\Sm$. Let $\alpha(t,r)\in [0,1]$ denote the maximal possible coefficient of the map $\Vm$ in such a convex combination. Then we can write
\begin{align}
&e^{-t(r+m)} \Rm\circ\sum^\infty_{k=0} \frac{t^k}{k!}\lb\Tm^{\Ltimes m} + r\Rm\rb^k\circ \mathcal{V}  =:\alpha\lb t,r\rb\Vm + \lb 1-\alpha\lb t,r\rb\rb\Sm.
\label{equ:alpha}
\end{align} 

We will show that $\alpha(t,r)\ra 1$ for any fixed time $t\in\R^+$ as $r\ra \infty$, thereby ensuring high-fidelity recovery. A lower bound on $\alpha\lb t,r\rb$ is obtained from the sum in (\ref{equ:alpha}) by only considering the contribution of terms in the expansion of each $\lb\Tm^{\Ltimes m} + r\Rm\rb^k$ such that no two $\Tm^{\Ltimes m}$ act right after another with no $\Rm$ in between to correct the error. Note that conditions (\ref{equ:Recovery2}) and (\ref{equ:ChannelErrCorr2}) guarantee that the recovery operation introduces no error when repeatedly applied to some encoded state $\mathcal{V}\lb\rho\rb$. And the quantum channel $\Tm^{\Ltimes m}$ is corrected by $\Rm$ when it acted only once on encoded state $\mathcal{V}\lb\rho\rb$. 

Using a combinatorial argument we can compute this contribution and obtain  
\begin{align*}
\alpha\lb t,r\rb &\geq e^{-t(r+m)}\sum^\infty_{k=0} \frac{t^k}{k!}\sum^{k}_{l=0} m^{k-l}r^l\binom{l+1}{k-l} \\
&=: f\lb mt,\frac{r}{m}\rb
\end{align*}
where we defined the lower bound $f:\R^+\times \R^+\ra\R$, which can be computed explicitly
\begin{align*}
f(t,r) =e^{-t(\frac{r}{2}+1)}\frac{a(r) \cosh\lb\frac{1}{2}a(r)t\rb + (2+r)\sinh\lb\frac{1}{2}a(r)t\rb}{a(r)} 
\end{align*} 
for $a(r)=\sqrt{r\lb 4+r\rb}$. See Fig. \ref{fig:alpha} for a plot of $f(t,r)$. It can be checked that indeed $f\lb mt,\frac{r}{m}\rb\ra 1$ for any fixed $t\in\R^+$ as $r\ra \infty$. This finishes the argument and shows that the rate $\frac{n}{m}$ is achievable, as we encoded $n$ qubits into $m$ qudits, for arbitrary $t\in \R^+$ by using the continuous coding scheme presented here. Note that we set $n_\nu := n$ and $m_\nu := m$ for all $\nu$ in Definition \ref{defn:ContQuantCap}.

\begin{figure}
     \centering
     \includegraphics[width=0.7\textwidth]{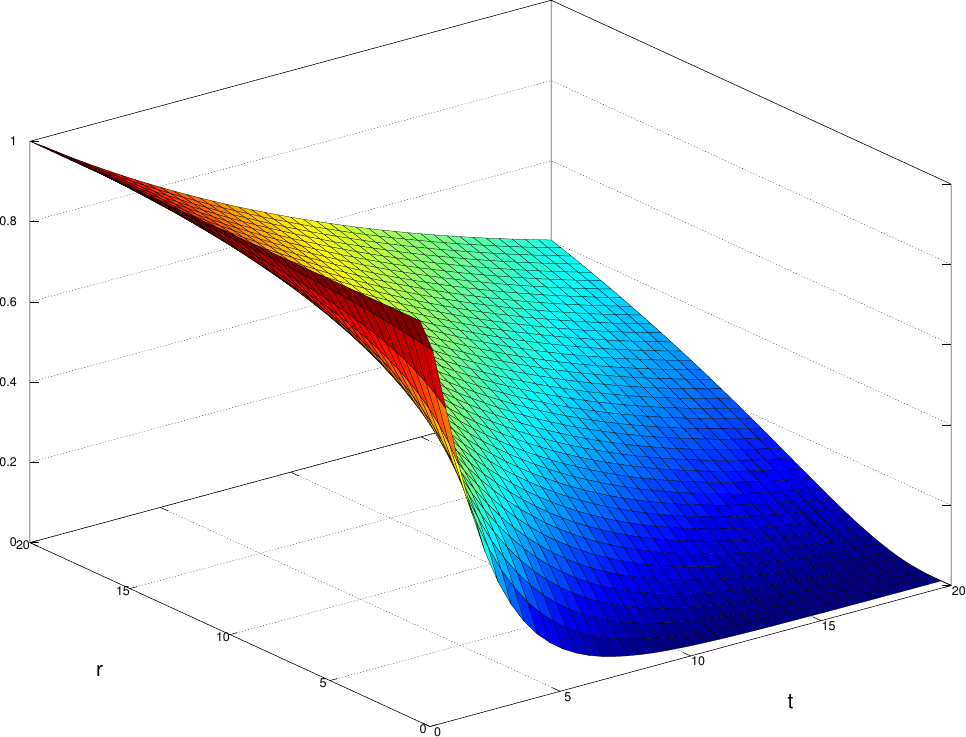}
     \caption{Plot of $f(t,r)$. This is a lower bound on the probability of recovering the information encoded in the continuously implemented quantum error correcting code as explained in the main text. For any time $t$ this probability can be made arbitrarily close to $1$ by choosing the rate $r$ of the coding Liouvillian high enough.}
     \label{fig:alpha}
\end{figure}

\end{proof}
\newpage
We will finish this section with two examples illustrating applications of the continuous coding scheme from the previous proof.

\begin{exmp}[Continuous coding schemes]\hfill
\begin{enumerate}
\item Consider the qubit depolarizing channel $\Tm_\text{dep}$ introduced in (\ref{equ:DepChannel}). We define the noise Liouvillian $\Lm:\M_2\ra \M_2$ as $\Lm = \Tm_{\text{dep}} - \i_2$. It is easy to see that the dynamical semigroup generated by the Liouvillian $\Lm$ has the form 
\begin{align*}
e^{t\Lm}\lb\rho\rb &=(1-e^{-\frac{4t}{3}})\text{tr}\lb\rho\rb \frac{\mathbbm{1}_2}{2} + e^{-\frac{4t}{3}}\i_2 \\
&\ra \text{tr}\lb \rho\rb\frac{\mathbbm{1}_{2}}{2}\text{ as }t\ra\infty.
\end{align*} 
For $t\geq \frac{3}{2}$, the channel $e^{t\Lm}$ remains completely positive after concatenating with the transposition map, which implies $\mathcal{Q}\lb e^{t\Lm}\rb = 0$ for $t\geq \frac{3}{2}$~\cite{Kretschmann2004}. Thus, any information stored even in many copies of the system will be lost after this time. 

For 5 copies of the system the noise is generated by the Liouvillian
\begin{align*}
\Lm^{\Ltimes 5} = \Tm^{\Ltimes 5}_{\text{dep}} - 5 \i_{2^5}
\end{align*}
with  
\begin{align*}
\Tm^{\Ltimes 5}_{\text{dep}}\lb\rho\rb = \frac{1}{3}\sum^5_{i=1} X^{(i)}\rho X^{(i)}+Y^{(i)}\rho Y^{(i)}+Z^{(i)}\rho Z^{(i)} .
\end{align*}
As the errors introduced by $\Tm^{\Ltimes 5}_{\text{dep}}$ are only single qubit errors we can use the 5-qubit stabilizer code~\cite[pp. 468]{Nielsen2007} to correct them. Theorem \ref{thm:ExplicitCodingLiou} then gives $\mathcal{Q}^\text{cont}_\chanSet\lb t\Lm\rb \geq \frac{1}{5}$ for arbitrarily long times, thus improving over $\mathcal{Q}\lb e^{t\Lm}\rb$ for times $t\geq \frac{3}{2}$. 

\item The scheme presented in this section is not restricted to the quantum setup, but can also be used to do continuous error correction of classical channels. To illustrate this consider a system of 3 classical bits and a bit-flip error of one uniform randomly chosen bit. This setting defines a Markov chain with a transition matrix $T^{\Ltimes 3}\in\M_{8}$, where $T\in\M_2$ is the transition matrix
\begin{align*}
T=\begin{pmatrix}
0 & 1 \\ 1 & 0
\end{pmatrix}.
\end{align*} 
The continuous-time Markov chain at a time $t\in\R^+$ corresponding to this error has the transition matrix
\begin{align*}
e^{tL}:= e^{t\lb T^{\Ltimes 3}-3\mathbbm{1}_8\rb}
\end{align*}
where we introduced the intensity matrix $L:= T^{\Ltimes 3}-3\mathbbm{1}_8$, which corresponds to the Liouvillian in the classical setup. 

To define an error correcting code, note that any single bit-flip error on 3 bits all in the same state (this defines the codespace) can be corrected by a majority vote. Therefore we define an encoding of 1 'logical' bit into the three bit system by a repetition code $(0)\mapsto(000)$ and $(1)\mapsto(111)$. This defines an isometric matrix $V$. The recovery matrix implementing the majority vote, i.e., mapping $(000)\mapsto(000)$, $(001)\mapsto(000)$, $(011)\mapsto(111)$, etc., is denoted by $R\in\M_{8}$. 

This error correcting code fulfills the classical analogues of the conditions (\ref{equ:Recovery2}) and (\ref{equ:ChannelErrCorr2}) from above.

In the same way as for the quantum case we can define a coding intensity matrix
\begin{align*}
L_c := r(R-\mathbbm{1}_8)
\end{align*}
for a rate $r\in\R^+$.

The calculation in the proof of Theorem \ref{thm:ExplicitCodingLiou} works now the same way in the classical case and we obtain
\begin{align*}
R\circ e^{t(L+L_c)}\circ V = \alpha(t,r)V + (1-\alpha(t,r))S
\end{align*}  
for an appropriate matrix S. For fixed $t\in\R^+$ we have $\alpha(t,r)\ra 1$ as $r\ra \infty$. 

\end{enumerate}
\label{exmp:ContCod}
\end{exmp}

The above example also shows that the addition of purely dissipative terms to the noisy time-evolution can increase its usual quantum capacity. A time-independent Liouvillian is called purely dissipative~\cite{Wolf2008} if it can be represented in the form 
\begin{align}
\Lm\lb\rho\rb = \sum^N_{k=1}\lb A_k\rho A^\dagger_k - \frac{1}{2}A^\dagger_k A_k\rho - \frac{1}{2}\rho A^\dagger_k A_k\rb
\end{align}
using only trace-less Lindblad operators $A_k$ and no Hamiltonian part $H$, see (\ref{equ:Liouvillian}).

The Lindblad operators $A_k$ of the coding Liouvillian constructed from a stabilizer code as in Example \ref{exmp:ContCod} 1)\ can be chosen as Hermitian projectors onto the syndrome spaces multiplied by error correcting unitaries conditioned on the measured syndrome~\cite[pp. 453 and pp. 468]{Nielsen2007}. Note that these Lindblad operators are all trace-less except the projector $P$ onto the code space itself, where the corresponding unitary is the identity matrix. But as this Lindblad operator is Hermitian, it can be replaced by the trace-less operator $\lb P-\text{tr}\lb P\rb\frac{\mathbbm{1}}{d}\rb$ yielding the same Liouvillian~\cite{Wolf2008}. In this way, Example 6.1\ leads to the following observation:

\begin{cor}[Dissipation can improve the usual quantum capacity]

There exist time-independent Liouvillians $\Lm:\M_d\ra\M_d$ and $\Lm':\M_d\ra\M_d$, where $\Lm'$ is purely dissipative such that 
\begin{align*}
\mathcal{Q}\lb e^\Lm\rb < \mathcal{Q}\lb e^{\Lm + \Lm'}\rb.
\end{align*}

\end{cor}

We even provided an example with usual quantum capacity $\mathcal{Q}\lb e^\Lm\rb = 0$ and $0<\mathcal{Q}\lb e^{\Lm + \Lm'}\rb$. The Example 6.1(2.)\ shows that a similar statement as in the above corollary also holds for the classical Shannon capacity, because $\mathcal{C}\lb e^{tL^{\Ltimes 3}}\rb\ra 0$ is also true for the intensity matrix $L$ introduced there.

\section{Conclusion and open problems}
\label{sec:7}

We have introduced and investigated quantum subdivision capacities and continuous quantum capacities to quantify the optimal rates of quantum information transmission in coding scenarios where intermediate intervention is possible. These capacities are a natural generalization of the standard quantum capacity and they capture recent results from quantum error correction~\cite{Pastawski2011} showing that controlled dissipation can help for information storage or transmission. 

There are many ways to adjust these two capacities to various applications. Every choice of the sets $\chanSet$, from which the intermediate coding channels or the coding Liouvillians are taken, yields a new capacity, which might have different properties. Many open problems arise in this context. What are good choices for the sets $\chanSet$ in the sense that they lead to mathematically interesting and physically reasonable capacities? More specifically, is there a closed form expression for the subdivision capacity $\mathcal{Q}_\unitary$ where only unitary add-ons are allowed? To what extend can $\mathcal{Q}\lb e^{\Lm+\Lm'}\rb$ differ from $\mathcal{Q}\lb e^\Lm\rb$? Do similar phenomena occur for classical capacities or for zero-error capacities?

\section*{Acknowledgements}
\label{sec:8}

We thank Geza Giedke, Fernando Pastawski, Oleg Szehr and Andreas Winter for fruitful discussions and valuable suggestions. Finally we are grateful to the Isaac Newton Institute for Mathematical Sciences and to the Benasque Center for Science where part of this work has been conducted.

\appendix
\section{The decoupling approach to quantum coding}
\label{sec:Appendix}
In this appendix we will review the proof that for an arbitrary quantum channel $\Tm:\M_{d_A}\ra \M_{d_B}$ any rate $R < I^{\text{coh}}\lb \omega^{A'A} ,\i_{A'}\otimes \Tm^{A\ra B}\rb$ is achievable. This is a special case of the direct part of the coding Theorem \ref{thm:LSD} for the quantum capacity, where any rate less than the regularized coherent information is shown to be achievable. For the proofs in Section 4 we only need this weaker version. This techniques were originally used in ~\cite{Winter2007} and further developed in ~\cite{Dupuis2010}, where the subsequent arguments can be found in a more general version. The basic idea of the decoupling approach can be motivated by stating the following theorem. 

\begin{thm}[Information-disturbance tradeoff, see ~\cite{Kretschmann2006}]
\label{thm:ID-Tradeoff}

For any quantum channel $\Tm:\M_{d_A}\ra\M_{d_B}$ with Stinespring isometry $V:\C^{d_A}\ra\C^{d_E}\otimes \C^{d_B}$ there exists a quantum state $\sigma^E\in\D_{d_E}$ such that
\begin{align}
\frac{1}{4}\inf_\Dm\left\Vert\i_{d_A} - \Dm\circ \Tm\right\Vert^2_\diamond &\leq \left\Vert \Tm^c - \text{tr}\lb\cdot\rb\sigma^E\right\Vert_\diamond \leq 2\inf_\Dm \left\Vert \i_{d_A} - \Dm\circ \Tm\right\Vert^\frac{1}{2}_\diamond
\label{equ:ID-Tradeoff}
\end{align}
where the infima are taken over all quantum channels $\Dm:\M_{d_B}\ra \M_{d_A}$.
\end{thm} 

The information-disturbance tradeoff states that the disturbance introduced by the quantum channel $\Tm$ can be corrected by a decoding quantum channel $\Dm$ iff the complementary channel, i.e. the channel that describes the information flow to the environment, is completely forgetful, i.e. conserves no information about the input state.

The idea of the decoupling approach is to use random encodings $\Em_\nu$, see Definition \ref{defn:QuantCap}, to ensure that the complementary channel $\lb\Tm^{\otimes m_\nu}\circ \Em_\nu\rb^c$ is completely forgetful in the limit $\nu\ra \infty$. Then Theorem \ref{thm:ID-Tradeoff} ensures the existence of decoding maps $\Dm_\nu$ that complete the coding scheme in Definition \ref{defn:QuantCap}.

Before we state the decoupling theorem we have to define some entropic quantities, which go back to ~\cite{Renner2004,Renner2005} and have been further developed in ~\cite{Tomamichel2012}.

\begin{defn}[Quantum conditional min-entropy, see ~\cite{Dupuis2010}]\hfill
\label{defn:SmoothEntropy}

For a positive matrix $\rho^{AB}\in \M^+\lb\C^{d_A}\otimes\C^{d_B}\rb$ we define the \textbf{conditional min-entropy of A given B} as
\begin{align*}
&H_\text{min}\lb A|B\rb_\rho := -\log\lb \min\lset\text{tr}\lb\sigma\rb : \sigma\in \M_{d_B}, \sigma \geq 0, \rho^{AB}\leq \mathbbm{1}_A\otimes \sigma\rset\rb.
\end{align*}
For $\epsilon > 0$ we define the $\epsilon$-\textbf{smooth conditional min-entropy of A given B} as 
\begin{align*}
H^\epsilon_\text{min}\lb A|B\rb_\rho := \max_{\sigma^{AB}\in\mathfrak{B}\lb\rho,\epsilon\rb} H_\text{min}\lb A|B\rb_\sigma  
\end{align*}
where $\mathfrak{B}\lb\rho,\epsilon\rb := \lset \tilde{\rho}^{AB}\in\M^+\lb\C^{d_A}\otimes \C^{d_B}\rb : \text{tr}\lb \tilde{\rho}^{AB}\rb \leq 1, \sqrt{1-F\lb \rho, \tilde{\rho}\rb}\leq \epsilon\rset$ denotes the $\epsilon-$ball in fidelity distance around $\rho$ within the set of subnormalized positive matrices.
\end{defn}
We will need some properties of the entropies defined above. The first property is a lower bound: For any density matrix $\rho\in \D\lb\C^{d_A}\otimes \C^{d_B}\rb$ we have
\begin{align*}
H^\epsilon_{\text{min}}\lb A|B\rb_\rho \geq -\log\lb d_B\rb.
\end{align*}
This lower bound can be seen easily from Definition \ref{defn:SmoothEntropy} as $\rho^{AB}\leq \mathbbm{1}_A\otimes \mathbbm{1}_B$ holds for every quantum state $\rho^{AB}\in\D\lb\C^{d_A}\otimes\C^{d_B}\rb$.

The second property is also called the \textbf{fully quantum asymptotic equipartition property}~\cite{Tomamichel2009} and is stated in the following lemma: 

\begin{lem}[Fully quantum AEP, see ~\cite{Tomamichel2009}]
\label{lem:AEP}

For a density operator $\rho^{AB}\in \D\lb\C^{d_A}\otimes\C^{d_B}\rb$ and any $\epsilon >0$ there exists a sequence $\Delta\lb n,\rho,\epsilon\rb\ra 0$ as $n\ra \infty$ such that for all $n\geq \frac{8}{5}\log\lb \frac{2}{\epsilon^2}\rb$ we have
\begin{align*}
\frac{1}{n}H_\text{min}^\epsilon\lb A^n| B^n\rb_{\rho^{\otimes n}} \geq S\lb A|B\rb_\rho - \Delta\lb n, \rho, \epsilon\rb.
\end{align*}   
\end{lem}

With the above definitions we can state the following decoupling theorem:

\begin{thm}[Decoupling theorem, see ~\cite{Dupuis20102}]
\label{thm:Decoupling}

For a density operator $\rho^{RA}\in \D\lb\C^R\otimes \C^A\rb$ and a quantum channel $\Tm^{A\ra E} : \M_{d_A}\ra \M_{d_E}$ with Choi-matrix $\sigma^{A'E} = \lb\i_{A'}\otimes \Tm\rb\lb\omega^{A'A}\rb$ and an arbitrary $\epsilon>0$ we have
\begin{align*}
&\int_{\unitary\lb d_A\rb} \left\Vert\lb\i_R\otimes \lbr\Tm^{A\ra E}\circ \Um\rbr\rb\lb \rho^{RA}\rb - \rho^R\otimes \sigma^E\right\Vert_1 \hspace*{0.1cm} d\Um \leq 2^{-\frac{1}{2}H^\epsilon_\text{min}\lb A'|E\rb_{\sigma} -\frac{1}{2}H^\epsilon_\text{min}\lb A|R\rb_{\rho}} + 12\epsilon
\end{align*}
where the integration is done wrt the Haar measure on the group of unitary maps $U:\C^{d_A}\ra\C^{d_A}$, which define unitary quantum channels $\Um:\M_{d_A}\ra \M_{d_A}$ via $\Um\lb\rho\rb = U\rho U^\dagger$, 
\end{thm} 
 
For further decoupling theorems in different scenarios and stronger versions of the above theorem see ~\cite{Dupuis2010,Szehr2011,Szehr2013}. 

Note that Theorem \ref{thm:Decoupling} implies the existence of a unitary channel $\Um:\M_{d_A}\ra\M_{d_A}$, which achieves the bound stated in the theorem. Applying this result to the complementary channel $\Tm^c:\M_{d_A}\ra \M_{d_E}$ of a given quantum channel $\Tm:\M_{d_A}\ra \M_{d_B}$ gives a bound on the second norm in equation (\ref{equ:ID-Tradeoff}) in Theorem \ref{thm:ID-Tradeoff} for the encoded complementary channel $\Tm^c\circ \Um$. Here we used that the diamond norm of a linear map $\Lm:\M_{d_A}\ra\M_{d_B}$ fulfills $\left\Vert\Lm\right\Vert_\diamond = \left\Vert\i_{R}\otimes \Lm\right\Vert_{1\ra 1}$ for a reference system $\M_{d_R}$ with $d_R = d_A$. 
 
In order to prove the achievability of a rate $R < I^{\text{coh}}\lb \omega^{A'A} ,\i_{A'}\otimes \Tm^{A\ra B}\rb$, we apply the decoupling theorem to many copies of the complementary channel. 
We have to show that the bound on the second norm in equation (\ref{equ:ID-Tradeoff}) vanishes in the limit of arbitrarily many copies. To obtain a vanishing bound using Theorem \ref{thm:Decoupling} one has to introduce a further isometry embedding the smaller system which is to be transmitted through the quantum channel into the larger system on which many copies of the channel are applied. By doing so, one can prove the following lemma:

\begin{lem}[Encoding operations, see ~\cite{Dupuis2010}]
\label{lem:RateAchievDecoup}

For a quantum channel $\Tm:\M_{d_A}\ra \M_{d_B}$ with corresponding complementary channel $\Tm^{c}:\M_{d_A}\ra \M_{d_E}$ and any rate $R < I^{\text{coh}}\lb \omega^{A'A} ,\i_{A'}\otimes \Tm^{A\ra B}\rb$ there exists a sequence $\epsilon_\nu$ which fulfills $\epsilon_\nu\ra 0$ as $\nu\ra \infty$ such that for any quantum state $\rho^{R'R}\in \D\lb\C^{d_{R'}}\otimes \C^{d_R}\rb$ with $d_{R'} = d_{R} = 2^{\nu R}$ 
\begin{align*}
\int_{\unitary\lb d^\nu_A\rb} \bigg{\Vert}& \lb\i_{R'}\otimes \lbr\lb\Tm^{c}\rb^{\otimes\nu}\circ \Um\circ \mathcal{V}^{R\ra A}_\nu\rbr\rb\lb \rho^{R'R}\rb - \rho^{R'}\otimes \lb\sigma^E\rb^{\otimes \nu}\bigg{\Vert} _1 \hspace*{0.1cm} d\Um \leq \epsilon_\nu
\end{align*}
holds. Here $\sigma^{A'E} := \lb\i_{A'}\otimes {\Tm^{c}}^{\otimes\nu}\rb\lb\omega^{A'A}\rb$ denotes the Choi matrix of the complementary quantum channel and for each $\nu\in\N$ we denote by $\mathcal{V}^{R\ra A}_\nu :\M_{d_{R}}\ra\M_{d^\nu_A}$ an arbitrary isometry.
\end{lem}

\begin{proof}
Consider the state $\psi^{R'A} = \lb\i_{R'}\otimes \mathcal{V}^{R\ra A}_\nu\rb\lb\rho^{R'R}\rb$. By applying Lemma \ref{thm:Decoupling} to $\psi^{R'A}$ we obtain for any $\epsilon>0$ 
\begin{align*}
\int_{\unitary\lb d^\nu_A\rb} \bigg{\Vert}& \lb\i_{R'}\otimes \lbr\lb\Tm^{c}\rb^{\otimes\nu}\circ \Um\circ \mathcal{V}^{R\ra A}_\nu\rbr\rb\lb \rho^{R'R}\rb - \rho^{R'}\otimes \lb\sigma^E\rb^{\otimes \nu}\bigg{\Vert}_1 \hspace*{0.1cm} d\Um \\
&\leq 2^{-\frac{1}{2}H^\epsilon_\text{min}\lb {A'}^\nu|E^\nu\rb_{\sigma^{\otimes\nu}} -\frac{1}{2}H^\epsilon_\text{min}\lb A|R'\rb_{\psi}} + 12\epsilon.
\end{align*}
By Lemma \ref{lem:AEP} we can estimate for sufficiently large $\nu$
\begin{align*}
&H^\epsilon_\text{min}\lb {A'}^\nu|E^\nu\rb_{\sigma^{\otimes\nu}} \\
&\geq \nu \lbr S\lb A'|E\rb_{\l|\sigma^{A'BE}\rk\lk\sigma^{A'BE}\r|} - \Delta\lb\nu,\sigma,\epsilon\rb\rbr \\
&= \nu\lbr I^{\text{coh}}\lb \omega^{A'A} ,\i_{A'}\otimes \Tm^{A\ra B}\rb - \Delta\lb \nu,\sigma,\epsilon\rb\rbr
\end{align*} 
where we introduced a purification $\l|\sigma^{A'BE}\rk\lk\sigma^{A'BE}\r|\in\M_{d_{A'}}\otimes\M_{d_{B}}\otimes\M_{d_E}$ of the Choi matrix $\sigma^{A'E}$ of the complementary channel and used that it is also a purification for the Choi matrix $\lb\i_{A'}\otimes \Tm^{A\ra B}\rb(\omega^{A'A})$ of the channel itself up to a local isometry. By the lower bound stated before, we have 
\begin{align*}
H^\epsilon_\text{min}\lb A|R'\rb_{\psi} \geq -\log\lb d_{R'}\rb = -\nu R.
\end{align*}
Combining the two bounds leads to 
\begin{align*}
\int_{\unitary\lb d^\nu_A\rb} \bigg{\Vert}& \lb\i_{R'}\otimes \lbr\lb\Tm^{c}\rb^{\otimes\nu}\circ \Um\circ \mathcal{V}^{R\ra A}_\nu\rbr\rb\lb \rho^{R'R}\rb - \rho^{R'}\otimes \lb\sigma^E\rb^{\otimes \nu}\bigg{\Vert}_1 \hspace*{0.1cm} d\Um \\
&\hspace*{-0.5cm}\leq 2^{\frac{1}{2}\nu\lbr R - I^{\text{coh}}\lb \omega^{A'A} ,\i_{A'}\otimes \Tm^{A\ra B}\rb + \Delta\lb\nu,\sigma,\epsilon\rb \rbr} + 12\epsilon.
\end{align*}
This proves the theorem for 
\begin{align*}
\epsilon_\nu = 2^{\frac{1}{2}\nu\lbr R - I^{\text{coh}}\lb \omega^{A'A} ,\i_{A'}\otimes \Tm^{A\ra B}\rb + \Delta\lb\nu,\sigma,\epsilon\rb \rbr} + \frac{12}{\nu}
\end{align*}
for the choice $\epsilon = \frac{1}{\nu}$.

\end{proof}

With proving the above theorem we are finished, because Theorem \ref{thm:ID-Tradeoff} guarantees the existence of decoding quantum channels such that the limit in equation (\ref{equ:limitQuantCap}) holds. By the above argument we proved that any rate $R < I^{\text{coh}}\lb \omega^{A'A} ,\i_{A'}\otimes \Tm^{A\ra B}\rb$ is achievable. 

The general form of the decoding operation can also be obtained more directly ~\cite{Dupuis2010}, by using Uhlmann's theorem~\cite{Uhlmann1976}. For later convenience we state the following

\begin{lem}[see ~\cite{Dupuis2010}]
\label{lem:GeneralDecoding}

Let $\Tm:\M_{d_A}\ra \M_{d_B}$ denote a quantum channel, with complementary channel $\Tm^{c}:\M_{d_A}\ra \M_{d_E}$ and $\epsilon >0$. If there exists a quantum state $\sigma^E\in\D_{d_E}$ such that for any quantum state $\rho^{RA}\in\D\lb\C^{d_R}\otimes \C^{d_A}\rb$ 
\begin{align*}
\left\Vert \lb\i_{R}\otimes \Tm^{c}\rb\lb \rho^{RA}\rb - \rho^{R}\otimes \lb\sigma^E\rb^{\otimes \nu}\right\Vert_1 \hspace*{0.1cm} \leq \epsilon
\end{align*}
is fulfilled, we can find a decoding quantum channel $\Dm:\M_{d_B}\ra \M_{d_A}$ such that for any quantum state $\rho^{RA}\in\D\lb\C^{d_R}\otimes \C^{d_A}\rb$
\begin{align*}
\left\Vert \lb \i_{d_R}\otimes \lbr\Dm\circ \Tm\rbr\rb\lb \rho^{RA} \rb - \rho^{RA}\right\Vert_1\leq 2\sqrt{\epsilon\lb 1-\frac{\epsilon}{4}\rb}.
\end{align*}
This decoding operation can be chosen to be of the form $D\lb\rho\rb = \text{tr}_{d_E}W\rho W^\dagger$, with an isometry $W:\C^{d_B}\ra\C^{d_A d_E}$.
\end{lem}

Note that the dimension $d_E$ of the environment is not the minimal dimension, which one might use for the representation of the decoding quantum channel. As the isometry $W:\C^{d_B}\ra\C^{d_A d_E}$ is obtained via relating two purifications of the involved states $\lb\i_{d_R}\otimes \Tm^{c}\rb\lb \l|\psi\rk\lk\psi\r|\rb$ and $\psi\otimes \sigma$ using Uhlmann's theorem, the minimal dimension $\tilde{d}_E$ for which there exists an isometry $\tilde{W}:\C^{d_B}\ra\C^{d_A \tilde{d}_E}$ such that an decoding operation can be guaranteed by Uhlmann's theorem is $\tilde{d}_E = \text{rank}\lb\sigma\rb$.

\bibliographystyle{IEEEtran}
\bibliography{mybibliography}

\end{document}